\documentclass[a4paper,USenglish,cleveref,pdfa]{lipics-v2021}

\usepackage{verbatim}

\pdfoutput=1
\hideLIPIcs
\nolinenumbers

\usepackage{preamble}

\title{Parameterized Max Min Feedback Vertex Set\footnote{An extended abstract of this work was presented at the
48th International Symposium on Mathematical Foundations of Computer Science (MFCS 2023)~\cite{mfcs/LampisMV23}.}}

\author{Michael Lampis}
{Universit\'{e} Paris-Dauphine, PSL University, CNRS UMR7243, LAMSADE, Paris, France}
{michail.lampis@dauphine.fr}
{https://orcid.org/0000-0002-5791-0887}{}

\author{Nikolaos Melissinos}
{Department of Theoretical Computer Science, Faculty of Information Technology, Czech Technical University in Prague, Czech Republic}
{nikolaos.melissinos@fit.cvut.cz}
{https://orcid.org/0000-0002-0864-9803}
{Supported by the CTU Global postdoc fellowship program.}

\author{Manolis Vasilakis}
{Universit\'{e} Paris-Dauphine, PSL University, CNRS UMR7243, LAMSADE, Paris, France}
{emmanouil.vasilakis@dauphine.eu}
{https://orcid.org/0000-0001-6505-2977}{}

\authorrunning{M. Lampis, N. Melissinos, and M. Vasilakis}

\Copyright{Michael Lampis, Nikolaos Melissinos, and Manolis Vasilakis}

\ccsdesc[500]{Theory of computation~Parameterized complexity and exact algorithms}%TODO mandatory: Please choose ACM 2012 classifications from https://dl.acm.org/ccs/ccs_flat.cfm 

\keywords{ETH, Feedback vertex set, Parameterized algorithms, Treewidth}

\category{} %optional, e.g. invited paper

\relatedversiondetails[linktext={https://arxiv.org/abs/2302.09604}]{Full Version}{https://arxiv.org/abs/2302.09604}

\funding{This work is partially supported by ANR projects ANR-21-CE48-0022 (S-EX-AP-PE-AL) and ANR-18-CE40-0025-01 (ASSK).}

\acknowledgements{Work primarily conducted while Nikolaos Melissinos was affiliated with Universit\'{e} Paris-Dauphine.}

\EventEditors{J\'{e}r\^{o}me Leroux, Sylvain Lombardy, and David Peleg}
\EventNoEds{3}
\EventLongTitle{48th International Symposium on Mathematical Foundations of Computer Science (MFCS 2023)}
\EventShortTitle{MFCS 2023}
\EventAcronym{MFCS}
\EventYear{2023}
\EventDate{August 28 to September 1, 2023}
\EventLocation{Bordeaux, France}
\EventLogo{}
\SeriesVolume{272}
\ArticleNo{60}
\begin{document}

\maketitle

\begin{abstract}
Given a graph $G$ and an integer $k$, \textsc{Max Min FVS} asks whether there exists a
minimal set of vertices of size at least $k$ whose deletion destroys all cycles.
We present several results that improve upon the state of the art of
the parameterized complexity of this problem with respect to both structural
and natural parameters.

Using standard DP techniques, we first present an algorithm of time $\textrm{tw}^{O(\textrm{tw})}n^{O(1)}$,
significantly generalizing a recent algorithm of Gaikwad et al.~of time $\textrm{vc}^{O(\textrm{vc})}n^{O(1)}$,
where $\textrm{tw}, \textrm{vc}$ denote the input graph's treewidth and vertex cover respectively.
Subsequently, we show that both of these algorithms are essentially optimal,
since a $\textrm{vc}^{o(\textrm{vc})}n^{O(1)}$ algorithm would refute the ETH.

With respect to the natural parameter $k$, the aforementioned recent work
by Gaikwad et al. claimed an FPT branching algorithm with complexity
$10^k n^{O(1)}$.
We point out that this algorithm is incorrect and present a
branching algorithm of complexity $9.34^k n^{O(1)}$.
\end{abstract}

\section{Introduction}

We consider a MaxMin version of the well-studied \textsc{Feedback Vertex Set} problem where,
given a graph $G=(V,E)$ and a target size $k$, we are asked to find a
set of vertices $S$ with the following properties: (i) every cycle of $G$
contains a vertex of $S$, that is, $S$ is a feedback vertex set,
(ii) no proper subset of $S$ is a feedback vertex set, that is, $S$ is minimal,
and (iii) $|S| \geq k$.
Although much less studied than its minimization cousin, {\mmFVS} has
recently attracted attention in the literature as part of a broader study of
MaxMin versions of standard problems, such as \textsc{Maximum
Minimal Vertex Cover} and \textsc{Upper Dominating Set}.
The main motivation of this line of research is the search
for a deeper understanding of the performance of simple greedy algorithms:
given an input, we would like to compute what is the worst possible solution
that would still not be improvable by a simple heuristic, such as removing
redundant vertices.
Nevertheless, over recent years MaxMin problems have been
found to possess an interesting combinatorial structure of their own and have
now become an object of more widespread study (we survey some such results below).

It is not surprising that {\mmFVS} is known to be NP-complete and is in fact
significantly harder than \textsc{Minimum FVS} in most respects, such as its
approximability or its amenability to algorithms solving special cases. Given
the problem's hardness, in this paper we focus on the parameterized complexity
of \mmFVS, since parameterized complexity is one of the main tools for dealing
with computational intractability.%
\footnote{Throughout the paper we assume that
the reader is familiar with the basics of parameterized complexity, as given in
standard textbooks~\cite{CyganFKLMPPS15}.}
We consider two types of parameterizations: the natural parameter $k$ and
the parameterization by structural width measures, such as treewidth.
In order to place our results into perspective, we first recall the current state of the art.

\subparagraph*{Previous work.} {\mmFVS} was first shown to be NP-complete even on
graphs of maximum degree $9$ by Mishra and Sikdar~\cite{MishraS01}. This was
subsequently improved to NP-completeness for graphs of maximum degree $6$ by Dublois et
al.~\cite{jcss/DubloisHGLM22}, who also present an approximation algorithm with
ratio $n^{2/3}$ and proved that this is optimal unless P=NP. A consequence of
the polynomial-time approximation algorithm of~\cite{jcss/DubloisHGLM22} was
the existence of a kernel of order $\bO(k^3)$, which implied that the problem is
fixed-parameter tractable with respect to the natural parameter $k$. Some
evidence that this kernel size may be optimal was later given by~\cite{AraujoBCS22}.
We note also that the problem can easily be seen to be FPT
parameterized by treewidth (indeed even by clique-width) as the property that a
set is a minimal feedback vertex set is MSO$_1$-expressible, so standard
algorithmic meta-theorems apply.

Given the above, the state of the art until recently was that this problem was
known to be FPT for the two most well-studied parameterizations (by $k$ and by
treewidth), but concrete FPT algorithms were missing. An attempt to advance
this state of the art and systematically study the parameterized complexity of
the problem was recently undertaken by Gaikwad et
al.~\cite{arxiv/GaikwadKMST22}, who presented exact algorithms for this problem
running in time $10^k n^{\bO(1)}$ and $\vc^{\bO(\vc)} n^{\bO(1)}$, where $\vc$ is the
input graph's vertex cover, which is known to be a (much) more restrictive
parameter than treewidth. Leveraging the latter algorithm,
\cite{arxiv/GaikwadKMST22} also presents an FPT approximation scheme which can
$(1-\varepsilon)$-approximate the solution in time $2^{\bO(\vc / \varepsilon)} n^{\bO(1)}$,
that is, single-exponential time with respect to $\vc$.

\subparagraph*{Our contribution.} We begin our work by considering {\mmFVS}
parameterized by the most standard structural parameter, treewidth.
We observe that, using standard DP techniques,
we can obtain an algorithm running in time $\tw^{\bO(\tw)} n^{\bO(1)}$, that is,
slightly super-exponential with respect to treewidth.
Note that this slightly super-exponential running time is already present in the
$\vc^{\bO(\vc)} n^{\bO(1)}$ algorithm of~\cite{arxiv/GaikwadKMST22}, despite the
fact that vertex cover is a much more severely restricted parameter.
Hence, our algorithm generalizes the algorithm of~\cite{arxiv/GaikwadKMST22} without a
significant sacrifice in the running time.

Despite the above, our main contribution with respect to structural parameters
is not our algorithm for the parameter treewidth, but an answer to a question that
is naturally posed given the above: can the super-exponential dependence
present in both our algorithm and the algorithm of~\cite{arxiv/GaikwadKMST22}
be avoided, that is, can we obtain a $2^{\bO(\tw)} n^{\bO(1)}$ algorithm?  We show
that this is likely impossible, as the existence of an algorithm running in
time $\vc^{o(\vc)} n^{\bO(1)}$ is ruled out by the ETH (and hence also the
existence of a $\tw^{o(\tw)} n^{\bO(1)}$ algorithm). This result is likely to be
of wider interest to the parameterized complexity community, where one of the
most exciting developments of the last fifteen years has arguably been the
development of the Cut\&Count technique (and its variations). One of the
crowning achievements of this technique is the design of single-exponential
algorithms for connectivity problems -- indeed an algorithm running in time
$3^{\tw} n$ for \textsc{Minimum FVS} is given in~\cite{CyganNPPRW22}. It has
therefore been of much interest to understand which connectivity problems admit
single-exponential algorithms using such techniques (see e.g.~\cite{BergougnouxBBK20}
and the references within). Curiously, even though
several cousins of \textsc{Minimum Feedback Vertex Set} have been considered in
this context (such as \textsc{Subset Feedback Vertex Set} and
\textsc{Restricted Edge-Subset Feedback Edge Set}~\cite{BergougnouxBBK20}),
for \mmFVS, which is
arguably a very natural variant, it was not known whether a single-exponential
algorithm for the parameter treewidth is possible. Our work thus adds to the
literature a natural connectivity problem where Cut\&Count can provably not be
applied (under standard assumptions).  Interestingly, our lower bound even
applies to the case of vertex cover, which is rare, as most problems tend to
become rather easy under this very restrictive parameter.
 
We then move on to consider the parameterization of the problem by $k$, the
size of the sought solution. Observe that a $k^{\bO(k)} n^{\bO(1)}$ algorithm can
easily be obtained by the results sketched above and a simple win/win argument:
starting with any minimal feedback vertex set $S$ of the given graph $G$, if
$|S| \ge k$ we are done; if not, then $\tw(G) \le k$ and we can solve the problem
using the algorithm for treewidth. It is therefore only interesting to consider
algorithms with a single-exponential dependence on $k$. Such an algorithm, with
complexity $10^k n^{\bO(1)}$, was claimed by~\cite{arxiv/GaikwadKMST22}.
Unfortunately, as we explain in detail in \cref{sec:natural}, this algorithm
contains a significant flaw.%
\footnote{Saket Saurabh, one of the authors of~\cite{arxiv/GaikwadKMST22},
confirmed so via private communication with Michael Lampis.}

Our contribution is to present a corrected version of the algorithm of~\cite{arxiv/GaikwadKMST22},
which also achieves a slightly better running time
of $9.34^k n^{\bO(1)}$, compared to the $10^k n^{\bO(1)}$ of the (flawed) algorithm
of~\cite{arxiv/GaikwadKMST22}. Our algorithm follows the same general strategy
of~\cite{arxiv/GaikwadKMST22}, branching and placing vertices in the forest or
the feedback vertex set. However, we have to rely on a more sophisticated
measure of progress, because simply counting the size of the selected set is
not sufficient. We therefore measure our progress towards a restricted special
case we identify, namely the case where the undecided part of the graph induces
a linear forest. Though this special case sounds tantalizingly simple, we show
that the problem is still NP-complete under this restriction, but obtaining an
FPT algorithm is much easier. We then plug in our algorithm to a more involved
branching procedure which aims to either reduce instances into this special
case, or output a certifiable minimal feedback vertex set of the desired size.

Finally, motivated by the above we note that a blocking point in the design of
algorithms for {\mmFVS} seems to be the difficulty of the extension problem:
given a set $S$, decide if a minimal fvs $S^*$ that extends $S$ exists.
Casel et al.~\cite{CaselFGMS22} showed that this problem is W[1]-hard parameterized by $|S|$.
Intriguingly, however, it is not even known if this problem is in XP, that is,
whether it is solvable in polynomial time for fixed $k$. We show that this is
perhaps not surprising, as obtaining a polynomial-time algorithm in this case would imply the existence of a
polynomial-time algorithm for the notorious $k$-\textsc{in-a-Tree} problem:
given $k$ terminals in a graph, find an induced tree that contains them. Since
this problem was solved for $k=3$ in a breakthrough by Chudnovsky and
Seymour~\cite{ChudnovskyS10}, the complexity for fixed $k\ge 4$ has remained a
big open problem (for example~\cite{LaiLT20} states that ``Solving it in
polynomial time for constant $k$ would be a huge result''). It is therefore
perhaps not surprising that obtaining an XP algorithm for the extension problem
for minimal feedback vertex sets of fixed size is challenging, since such an
algorithm would settle another long-standing problem.

\subparagraph*{Other relevant work.} As mentioned, {\mmFVS} is an example of a wider
class of MaxMin problems which have recently attracted much attention in the
literature, among the most well-studied of which are \textsc{Maximum
Minimal Vertex Cover}~\cite{AraujoBCS22,BonnetLP18,dam/BoriaCP15,siamdm/Zehavi17} and \textsc{Upper
Dominating Set} (which is the standard name for \textsc{Maximum Minimal
Dominating Set})~\cite{AbouEishaHLMRZ18,AraujoBCS23,BazganBCFJKLLMP18,DubloisLP22}.
Besides these problems, MaxMin or MinMax versions of cut and separation problems~\cite{DuarteEHKKLPSS21,HanakaKKY21,Lampis21},
knapsack problems~\cite{FuriniLS17,GourvesMP13},
matching problems~\cite{ChaudharyMP23}, and
coloring problems~\cite{BelmonteKLMO22} have also been studied.

The question of which connectivity problems admit single-exponential algorithms
parameterized by treewidth has been well-studied over the last decade. As
mentioned, the main breakthrough was the discovery of the Cut\&Count technique~\cite{CyganFKLMPPS15},
which gave randomized $2^{\bO(\tw)} n^{\bO(1)}$ algorithms
for many such problems, such as \textsc{Steiner Tree}, \textsc{Hamiltonicity},
\textsc{Connected Dominating Set}, and others. Follow-up work also provided
deterministic algorithms with complexity $2^{\bO(\tw)} n^{\bO(1)}$~\cite{BodlaenderCKN15}.
It is important to note that the discovery of these
techniques was considered a surprise at the time, as the conventional wisdom
was that connectivity problems probably require $\tw^{\bO(\tw)}$ time to be
solved~\cite{LokshtanovMS18}.
Naturally, the topic was taken up with much
excitement, in an attempt to discover the limits of such techniques, including
problems for which they cannot work. In this vein, \cite{Pilipczuk11} gave a
meta-theorem capturing many tractable problems, and also an example problem
that cannot be solved in time $2^{o(\tw^2)} n^{\bO(1)}$ under the ETH. Several
other examples of connectivity problems which require slightly
super-exponential time parameterized by treewidth are now known~\cite{BasteST20,HarutyunyanLM21},
with the most relevant to our work being
the feedback vertex set variants studied in~\cite{BergougnouxBBK20,BonnetBKM19},
as well as the digraph version of the
minimum feedback vertex set problem (parameterized by the treewidth of the underlying graph)~\cite{BonamyKNPSW18}.
The results of our paper seem to
confirm the intuition that the Cut\&Count technique is rather fragile when
applied to feedback vertex set problems, since in many variations or
generalizations of this problem, a super-exponential dependence on treewidth is
inevitable (assuming the ETH).

\section{Preliminaries} \label{sec:prelim}
Throughout the paper we use standard graph notation~\cite{Diestel17}.
A \emph{multigraph} $G$ is a graph which is permitted to have multiple edges with the same end nodes,
thus, two vertices may be connected by more than one edge.
Given a (multi)graph $G$, where $e = \{u,v\} \in E(G)$ is an edge connecting distinct vertices $u$ and $v$,
the \emph{contraction} of $e$ results in a new graph $G' = G / uv$ such that $V(G') = (V(G) \setminus \braces{u,v}) \cup \braces{w}$,
while for each edge $\braces{u,x}$ or $\braces{v,x}$ in $E(G)$ with $x \in V(G) \setminus \{u,v\}$,
there exists an edge $\braces{w,x}$ in $E(G')$.
Any edge $e \in E(G)$ not incident to $u,v$ also belongs to $E(G')$.
If $u$ and $v$ were additionally connected by an edge apart from $e$, then $w$ has a self loop.
Moreover, for vertex $u \in V(G)$,
let $\deg_{X} (u)$ denote its degree in $G[X \cup \braces{u}]$, where $X \subseteq V(G)$.
Additionally, given a subgraph $G'$ of $G$, $N_{G'}(u)$ denotes the \emph{neighborhood} of $u$ in $G'$,
where we omit the subscript when the graph is clear from the context.
A feedback vertex set $S$ of $G$ is \emph{minimal} if and only if $\forall s \in S$,
$G[(V(G) \setminus S) \cup \braces{s}]$ contains a cycle, namely a \emph{private cycle} of $s$~\cite{DubloisLP22}.

Let $\Z$ denote the set of all integers, while $\N$ denotes the set of positive integers.
For $x, y \in \Z$, let $[x, y] = \setdef{z \in \Z}{x \leq z \leq y}$,
while $[x] = [1,x]$.
We make use of a weaker version of the ETH,
which states that $3$-SAT cannot be determined in time $2^{o(n)}$,
where $n$ denotes the number of variables~\cite{jcss/ImpagliazzoPZ01}.
Lastly, we make use of the following theorem.

\begin{theorem}[{\cite[Theorem 3.10]{books/aw/GKP1994}}]\label{thm:ceilings}
    Let $f(x)$ be any continuous, monotonically increasing function with the
    property that if $f(x) \in \Z$, then $x \in \Z$.
    Then, $\ceil{f(\ceil{x})} = \ceil{f(x)}$ whenever $f(\ceil{x})$ and $f(x)$ are defined.
\end{theorem}

\section{Treewidth Algorithm}\label{sec:tw_algo}

We start by presenting an algorithm for {\mmFVS} parameterized by the treewidth of the input graph,
arguably the most well-studied structural parameter.
As a corollary of the lower bound established in \cref{sec:eth_vc_lb},
it follows that the running time of the algorithm is essentially optimal under the ETH.

\begin{theoremrep} \label{thm:fpt-tw-D}
    Given an instance $\mathcal{I} = (G, k)$ of \mmFVS, as well as a nice tree decomposition of $G$ of width \tw,
    there exists an algorithm that decides $\mathcal{I}$ in time $\tw^{\bO(\tw)} n^{\bO(1)}$.
\end{theoremrep}

\begin{proof}
The main idea lies in performing standard dynamic programming on the nodes of the nice tree decomposition.
For a node $t$ of the tree decomposition, let $B_t$ denote its bag,
and $B_t^{\downarrow} \supseteq B_t$ denote the union of the bags in the subtree rooted at $t$.

Let $S^* \subseteq V$ be a minimal feedback vertex set of $G$,
where $F^* = V \setminus S^*$ and $G[F^*]$ is a forest.
For each $u \in S^*$, it holds that there exists a set of vertices $T_u \subseteq F^*$ such that
$G[T_u]$ is a tree and $G[T_u \cup \{u\}]$ is not acyclic, as $u$ has a private cycle
containing at least two of its neighbors in $T_u$.
Our goal is, for each node $t$, to build all partial solutions $S$,
where $S \subseteq B_t^{\downarrow}$ is a (not necessarily minimal) feedback vertex set of $G[B_t^{\downarrow}]$ and
for each $u \in S \setminus B_t$, its neighboring vertices in its private cycle belong to $B_t^{\downarrow} \setminus S$. 
By considering all the partial solutions of the root node of the tree decomposition
we can determine a maximum minimal feedback vertex set of the input graph $G$.

We will store this information using a coloring on the vertices of the graph.
In particular, notice that for any minimal feedback vertex set $S^*$ of $G$ there exists a coloring
$C^* \colon V \to [\tw+1]$ of the vertices of $G$ such that the following hold.
\begin{itemize}
    \item For any pair $u,v \in V \setminus S^*$,
    if $u,v$ are in the same connected component of $G[V \setminus S^*]$, then $C^*[u] = C^*[v]$.
    
    \item For any pair $u,v \in V \setminus S^*$,
    if $u,v$ are in different connected components of $G[V \setminus S^*]$ and $C^*[u] = C^*[v]$,
    then for all nodes $t$ of the tree decomposition it holds that $\{u,v\} \not\subseteq B_t$
    (i.e., there is no bag of the tree decomposition that intersects two distinct connected components that share the same color).
    
    \item For every vertex $s \in S^*$, $s$ has a private cycle
    (i.e., a cycle that contains only vertices of $(V \setminus S^*) \cup \{ s \}$)
    where all vertices $u$ in this cycle have color $C^*[u]=C^*[s]$.
\end{itemize}
\noindent
Any such coloring $C^*$ will be called \emph{valid} for $S^*$.

Now consider a node $t$ of the tree decomposition and a (not necessarily minimal)
feedback vertex set $S$ of $G[B_t^\downarrow]$.
Given $S$ and $t$,
we will store only the colorings $C$ of $B_t^\downarrow$ that are \emph{extendable with respect to $S$ and $t$},
that is, those that could potentially be extended to
a valid coloring for some minimal feedback vertex set $S^*$ of $G$,
such that $S^* \supseteq S$ and $V \setminus S^* \supseteq B_t^\downarrow \setminus S$.

We proceed by identifying two properties of the colorings that are extendable with respect to a given set $S$ and $t$. 

\proofsubparagraph*{Property 1.}
Assume that we are restricted to the subtree $B_t^\downarrow$ of the tree decomposition rooted at node $t$,
and let $S = S^* \cap B_t^\downarrow$ denote the
restriction of a minimal feedback vertex set $S^*$ of $G$ to said subtree.
Now, let $C^*$ be a coloring of $G$ which is valid with respect to $S^*$ and let%
\footnote{In what follows,
we slightly abuse notation and use the vertex set of a connected component to refer to the component itself as well.}
$V_1, V_2$ be two connected components of $G[B_t^\downarrow \setminus S^*]$
(or equivalently, of $G[B_t^\downarrow \setminus S]$)
such that $C^*[v] = C^*[u]$ for all $v \in V_1$ and $u \in V_2$.
Then, one of the following holds:
\begin{itemize}
    \item either there is no node $t'$ in the subtree rooted at node $t$ such that $B_{t'} \cap V_1 \neq \varnothing$ and $B_{t'} \cap V_2 \neq \varnothing$,
    that is, $t'$ does not intersect both $V_1$ and $V_2$,
    \item or $B_t$ intersects both $V_1$ and $V_2$.
\end{itemize}
Notice that the second case can only be when the vertices of $V_1$ and $V_2$ belong to the same connected component of $G[V \setminus S^*]$
without that being the case for the subgraph $G[B_t^\downarrow \setminus S^*]$.

\proofsubparagraph*{Property 2.}
Now consider a minimal feedback vertex set $S^*$ of $G$ and 
a coloring $C$ of $B_t^\downarrow$ that is extendable,
with respect to $S = S^* \cap B^\downarrow_t$ and $t$, to a coloring $C^*$ which is in turn valid for $S^*$.
Property 1 implies that extendable colorings $C$
may color distinct connected components of $G[B_t^\downarrow]$ with the same color,
only if said components both intersect $B_t$.
Additionally, for any vertex $v \in S \setminus B_t$,
it holds that its two neighbors in its private cycle must be in $B_t^\downarrow$;
if that is not the case, $v$ cannot have a private cycle,
since it is not incident to any vertex of $V \setminus B_t^\downarrow$.
We will refer to these two neighbors of $v$ as \emph{interesting}.
Notice that both interesting neighbors of $v$ must have color $C[v]$ and
either be in the same connected component of $G[B_t^\downarrow \setminus S]$,
or in two distinct connected components of $G[B_t^\downarrow \setminus S]$ that both intersect $B_t$.

A triplet $(t,S,C)$ will be called a \emph{valid partial solution} if
(i) $S$ is a feedback vertex set of $G[B_t^\downarrow]$,
and
(ii) $C$ is an extendable coloring of $B_t^\downarrow$ with respect to $S$ and $t$. 
In that case,
by definition it holds that $S$ and $C$ have the potential to be restrictions
of a minimal feedback vertex set $S^*$ of $G$ and of a coloring $C^* \colon V \to [\tw+1]$ in $B^\downarrow_t$,
where $C^*$ is valid for $S^*$. 
Whenever we refer to a \emph{potential extension} of $S$ and $C$, or to the \emph{potential final forest},
we assume the existence of such $S^*$ and $C^*$ and refer to those, or to the forest $G[V \setminus S^*]$ respectively.

Notice that the discussion so far implies that the total number of valid partial solutions $(t,S,C)$
is not bounded by any function of $\tw$.
To cope with this issue,
we will define an FPT number of types of solutions and we will keep one pair $(t,S,C)$ per type.
In particular, we will keep the one of the largest feedback vertex set.

To this end, given a valid partial solution $(t,S,C)$ 
we define a function $D \colon S \cap B_t \to \{0,1,2\}$ and
a partition $\mathcal{F}$ of $B_t \setminus S$.
We note that given the partial solution $(t,S,C)$, both $D$ and $\mathcal{F}$ will be computable in polynomial time.

%\todo{\manolis{Changed the following, read and check}}
Regarding the function $D$, it indicates for each vertex $v \in S\cap B_t$ how many vertices
of $B_t^\downarrow \setminus S$ could be its neighbors in a potential private cycle of $v$.
In particular, we consider the following two cases.
In the first case, assume that there exist two vertices $v_1,v_2 \in N(v)$
belonging to the same connected component $U$ of $G[B_t^{\downarrow} \setminus S]$
with $C[v_1] = C[v_2] = C[v]$.
In that case, assuming there is a minimal feedback vertex set $S^*$ of
$G$ such that $S^* \supseteq S$ and $V \setminus S^* \supseteq B_t^\downarrow \setminus S$,
it holds that $v_1$ and $v_2$ are neighbors of $v \in S^*$ in its private cycle,
thus we set $D[v] = 2$.
In the second case, no such two vertices exist, that is,
every same-colored neighbor of $v$ in $B_t^{\downarrow} \setminus S$
belongs to a distinct connected component of $G[B_t^{\downarrow} \setminus S]$.
Then, due to Property 2, it suffices to set $D[v] = \min \{2, \ell\}$,
where $\ell$ is equal to the number of same-colored neighbors of $v$ in $B_t^{\downarrow} \setminus S$
that belong to a connected component $U$ of $G[B_t^\downarrow \setminus S]$
such that $U \cap B_t \neq \varnothing$.

As for the partition $\mathcal{F}$ of $B_t \setminus S$, it dictates the connectivity of the vertices
of $B_t \setminus S$ in the graph $G[B_t^\downarrow \setminus S]$.
In particular, we define $\mathcal{F} = \setdef{F_i}{i \in [q]}$,
with $q \le \tw+1$, to be a partition of
$B_t \setminus S$ such that
vertices $u$ and $v$ belong to the same set $F \in \mathcal{F}$ if and only if
they belong to the same connected component of $G[B_t^\downarrow \setminus S]$.

\proofsubparagraph*{Types.}
Consider two valid partial solutions $(t,S_1,C_1)$ and $(t,S_2,C_2)$.
Let $D_i$ and $\mathcal{F}_i$ be the table and partition defined from $(t,S_i,C_i)$, where $i \in [2]$.
We will say that $(t,S_1,C_1)$ and $(t,S_2,C_2)$ have the same type if:
\begin{itemize}
    \item $S_1 \cap B_t = S_2 \cap B_t$,
    \item $C_1[u] = C_2[u]$ for all $u \in B_t$,
    \item $D_1[u] = D_2[u]$ for all $u \in S_1 \cap B_t$,
    \item $\mathcal{F}_1$ and $\mathcal{F}_2$ define the same partition on the vertices of
    $B_t \setminus S_1$.
\end{itemize}
For any valid partial solution $(t,S,C)$,
we will say that $(t,S,C)$ is of type $(S_t,C_t, D_t, \mathcal{F}_t)$,
where $S_t = S\cap B_t$, $C_t$ is the restriction of $C$ in $B_t$,
while $D_t$ and $\mathcal{F}_t$ are the table and the partition as previously defined.
Notice that the type of any partial solution can be computed in $n^{\bO(1)}$ time.
In what follows, we assume that for any node $t$ of the tree decomposition,
when we build two valid partial solutions $(t,S_1,C_1)$ and $(t,S_2,C_2)$ of the same type,
we only keep the one with the larger feedback vertex set,
namely $(t,S_1,C_1)$ if $|S_1| \ge |S_2|$ and $(t,S_2,C_2)$ otherwise. 

We now proceed to explaining how to compute the partial solutions for each node of the tree decomposition.

\proofsubparagraph*{Leaf Nodes.}
Since the bags of Leaf Nodes are empty, it follows that we keep a partial solution $(t,S,C)$ where $t$ is the considered leaf node and 
both $S$ and $C$ are empty.
Note that this suffices for any partial solution we need to keep.

\proofsubparagraph*{Introduce Nodes.}
Let $t$ be an Introduce Node,
where $t'$ denotes its child node and $u$ the newly introduced vertex.

We build the valid partials solutions for $t$ as follows.
For each partial solution $(t',S', C')$ we extend $S'$ and $C'$ by considering all valid options
regarding the introduced vertex $u$. We discard any invalid partial solution.
In the following, let $C_i$ denote the extension of $C'$ to $B_{t'}^{\downarrow} \cup \{u\}$
where $C_i[u] = i$, for all $i \in [\tw+1]$.

First consider the case where we are extending $(t',S', C')$ by adding $u$ to the partial solution $S'$.
To this end, for each $i \in [\tw+1]$ we create a partial solution $(t,S' \cup \{u\}, C_i)$,
in which case the private cycle of $u$ is using vertices of color $i$.
Notice that the inclusion of $u$ to $S'$ does not create any invalid partial solutions.

Next, we extend $(t',S', C')$ by placing $u$ in the forest.
Once again, we will consider all different colorings of $u$.
Fix $i \in [\tw+1]$.
Note that, if there exists a vertex $v \in N(u) \cap (B_t \setminus S')$ such that $C_i[v] \neq i$,
then $(t,S', C_i)$ is an invalid partial solution,
as we use multiple colors for the same connected component of $G[B_t^{\downarrow} \setminus S']$, thus we discard it.
Furthermore,
if $u$ has at least two neighbors $v_1,v_2 \in N(u) \cap (B_t \setminus S')$ such that $v_1, v_2$ are
connected in $G[B_{t'}^{\downarrow} \setminus S']$, 
then $G[B^\downarrow_t \setminus S']$ contains a cycle, thus $(t,S', C_i)$ is an invalid partial solution.
If none of the above holds, then $(t,S',C_i)$ is a non-discarded valid partial solution.

\proofsubparagraph*{Join Nodes.}
Let $t$ be a Join Node and $t_1$ and $t_2$ denote its two children. 
First, we explain how we create the partial solutions for $t$ by
using the previously built partial solutions for nodes $t_1$ and $t_2$. 

Let $(t_1,S_1,C_1)$ and $(t_2,S_2,C_2)$ be two partial solutions stored for $t_1$ and $t_2$ respectively,
with $S_1 \cap B_{t_1} = S_2 \cap B_{t_2}$ and $C_1[v] = C_2[v]$ for all $v \in B_t$. 
Consider the partial solution $(t,S,C)$ where
\begin{itemize}
    \item $C[u] = C_1[u]$ if $u \in B_{t_1}^\downarrow$,
    \item $C[u] = C_2[u]$ if $u \in B_{t_2}^\downarrow$,
    \item and $S=(S_1 \cup S_2)$.
\end{itemize}
Note that $(t,S,C)$ is a valid partial solution if $B_t^{\downarrow} \setminus S$ is acyclic
and this can be checked in polynomial time; if that is not the case we discard it.

\proofsubparagraph*{Forget Nodes.}
Let $t$ be a Forget Node, where $t'$ denotes its child node and $u$ the forgotten vertex.
Again, we start by explaining how we create the partial solutions for $t$ by using
the previously built partial solutions for $t'$.
To this end, for each $(t',S',C')$ we have stored,
we create the partial solution $(t,S',C')$ and check whether it is valid.
To do so, we distinguish between the cases $u \in S'$ and $u \notin S'$.

If $u \in S'$, then we need to verify whether $u$ has found at least $2$ of its neighbors
which are included in its private cycle in the potential final solution.
This can be done in polynomial time by checking whether $D_{t'}[u] = 2$,
where $(S'_{t'},C'_{t'},D_{t'},\mathcal{F}_{t'})$ denotes the type of $(t',S',C')$;
if that is indeed the case, we keep $(t,S',C')$, otherwise we discard it.

Now we consider the case where $u \notin S'$.
Here, we need to check whether the vertices of $\setdef{v \in B_{t'}\setminus S'}{C'[v] = C'[u]}$
can still be in the same connected component of the potential final forest after the removal of $u$. 
We consider two cases, either for all $v \in B_t \setminus S'$ it holds that $C'[v] \neq C'[u]$,
or not.

In the first case, there is no vertex in $B_t$ that should be in the same connected component of the potential final forest as $u$. Therefore, we keep $(t,S',C')$.

In the latter, let $U \subseteq B_t^\downarrow \setminus S'$ be the connected component of
$G[B_t^\downarrow \setminus S']$ where $u \in U$.
It suffices to check whether $U \cap B_t = \varnothing$.
If $U \cap B_t = \varnothing$, we discard $(t,S',C')$ as $u$ cannot be in the same connected component as
the rest of the vertices of $\setdef{v \in B_{t'} \setminus S'}{C'[v] = C'[u]}$ in a potential extension of $(S',C')$.
If $U \cap B_t \neq \varnothing$, we store $(t,S',C')$ as it is a valid partial solution.

We now proceed to proving the correctness of the described algorithm.
In particular, we will show that in order to find a minimal feedback vertex set of $G$ of maximum size,
it suffices to keep for each node one partial solution per type.

\begin{lemma}\label{lemma:tw_algo_correctness}
    Let $(t,S,C)$ be a valid partial solution.
    The algorithm builds a valid partial solution $(t,\bar{S}, \bar{C})$ of the same type,
    where $|\bar{S}| \ge |S|$.
\end{lemma}

\begin{proof}
    Let $(t,S,C)$ be a valid partial solution and $(S_t,C_t,D_t,\mathcal{F}_t)$ its type.
    When we consider the restriction of $(t,S,C)$ in a child node $t'$ of $t$
    we refer to the partial solution $(t',S',C')$ where $S' = S \cap B_{t'}^\downarrow$ and
    $C' \colon B_{t'}^\downarrow \to [\tw+1]$ is the restriction of $C$ in $B_{t'}^\downarrow$.
    Notice that since $(t,S,C)$ is a valid partial solution,
    $(t',S',C')$ is also a valid partial solution.
    We will deal with each kind of node separately and prove the statement by induction.
    
    \proofsubparagraph*{Leaf Nodes.}
    Since the bags of Leaf Nodes are empty, the statement is trivially true.
    Notice that the leaf nodes correspond to the base case of the induction.

    \proofsubparagraph*{Introduce Nodes.}
    Let $t$ be an Introduce Node, where $t'$ is its child node and $u$ is the newly introduced vertex. 
    Assume that the claim holds for node $t'$.
    In that case, we have stored a partial solution $(t',\bar{S'},\bar{C'})$ of the same type as $(t',S', C')$ such that $|\bar{S'}|\ge |S'|$.
    We will prove that, when the algorithm extends $(t',\bar{S'},\bar{C'})$,
    it builds a valid partial solution $(t,\bar{S},\bar{C})$ that is of the same type as $(t,S,C)$ with $|\bar{S}| \ge |S|$.

    Initially consider the case where $u \in S$ and $C[u] = c$.
    Notice that the algorithm always creates a \emph{valid} partial solution $(t,\bar{S}, \bar{C})$ for node $t$ that
    extends $(t',\bar{S'},\bar{C'})$ and where $u \in \bar{S}$ and $\bar{C}[u] = c$.
    Let $(\bar{S}_t,\bar{C}_t, \bar{D}_t, \bar{\mathcal{F}}_t)$ denote the type of $(t,\bar{S} , \bar{C})$,
    in which case it suffices to prove that this is the same as the type of $(t,S,C)$,
    which is $(S_t,C_t,D_t,\mathcal{F}_t)$.
    Additionally,
    let $(S'_{t'},C'_{t'}, D'_{t'}, \mathcal{F}'_{t'})$ be the type of $(t',S',C')$ and $(t',\bar{S'},\bar{C'})$.

    First, observe that $S_t = S'_t \cup \{u\} = \bar{S}_t$,
    as well as that $C_t[v] = \bar{C}_t[v]$ for all $v \in B_t$
    (as they both extend $C'_{t'}$ and assign the color $c$ to $u$).    

    Now we deal with the partitions $\mathcal{F}_t$ and $\bar{\mathcal{F}}_t$.
    Notice that adding $u$ to the feedback vertex set does not affect the connectivity.
    Consequently, $\mathcal{F}_t$, $\mathcal{F}'_{t'}$, and $\bar{\mathcal{F}}_t$ define
    the same partition into connected components for the vertices of $B_{t'} \setminus (S \cup \{u\}) = B_t \setminus S$.

    Finally, we show that $D_t[v] = \bar{D}_t[v]$ for all $v \in S \cap B_t$.
    Notice that this holds for all $v \in S \cap B_{t'}$,
    as adding $u$ to the feedback vertex set does not affect the interesting neighbors of any other vertex.
    It remains to argue for vertex $u$,
    for which it holds that%
    \footnote{Recall that since $u$ is the introduced vertex,
    it has no neighbors in $B_t^\downarrow \setminus B_t$.}
    $D_t[u] = \min\{2, |N(u) \cap \setdef{v \in B_t \setminus S_t}{C[v] = c}|\}$ and
    $\bar{D}_t[u] = \min\{2, |N(u) \cap \setdef{v \in B_t \setminus S_t}{\bar{C}[v] = c}|\}$,
    and $D_t[u] = \bar{D}_t[u]$ follows from the fact that both $C$ and $\bar{C}$ agree in $B_t$.

    Consequently, the created partial solution $(t,\bar{S}, \bar{C})$ has the same type as $(t,S,C)$.

    Now we deal with the case where $u \notin S$ and $C[u]= c$.
    Notice that the algorithm always creates a (possibly invalid) partial solution
    $(t,\bar{S}, \bar{C})$ that extends $(t',\bar{S'},\bar{C'})$ such that
    \begin{itemize}
        \item $\bar{S} = \bar{S'}$ and 
        \item $\bar{C} \colon B_t^\downarrow \to [\tw+1]$,
        where $\bar{C}[v] = \bar{C'}[v]$ for all $v \in B_{t'}^\downarrow$ and $\bar{C}[u] = c$.
    \end{itemize}
    We claim that if $(t,\bar{S},\bar{C})$ is an invalid partial solution (and is thus discarded),
    then $(t,S,C)$ is also invalid.

    First, observe that $\bar{S} \cap B_t = \bar{S'} \cap B_t = S \cap B_t$,
    and $C_t[v] = \bar{C}_t[v]$ for all $v \in B_t$
    (as they both extend $C'_{t'}$ and assign the color $c$ to $u$). 
    Notice that, while the algorithm extends $(t',\bar{S'}, \bar{C'})$ to $(t,\bar{S},\bar{C})$,
    it discards the latter in two cases:
    \begin{enumerate}
        \item There exists a vertex $v \in B_t \setminus \bar{S}$ such that $v \in N(u)$ and $\bar{C}[v]\neq c$.
        \item The addition of $u$ to the potential forest creates a cycle in $G[B_t^\downarrow \setminus \bar{S}]$.
    \end{enumerate}
    We will show that neither of these may happen if $(t,S,C)$ is a valid partial solution.

    Assume that the first case holds, and there exists a vertex $v \in B_t \setminus \bar{S}$ such that
    $v \in N(u)$ and $\bar{C}[v] \neq c$.
    Since $S \cap B_t = \bar{S} \cap B_t$ and the colorings $C$ and $\bar{C}$
    agree in the vertices of $B_t$,
    it holds that $v \in B_t \setminus S$ with $v \in N(u)$ and $C[v] \neq c$.
    In that case, there exists a connected component of $G[B_t^\downarrow \setminus S]$ which is not monochromatic,
    thus $(t,S,C)$ cannot be a valid partial solution of $t$.

    In the second case, we will show that $G[B_t^\downarrow \setminus {S}]$ has a cycle
    which contradicts to the assumption that $(t,S,C)$ is a valid partial solution. 
    Let $u - u_1 - \ldots - u_\ell$ be a cycle in $G[B_t^\downarrow \setminus \bar{S}]$ that includes $u$,
    where $u_1, u_\ell \in B_t \setminus \{ u \} = B_{t'}$,
    since $u$ cannot have any neighbors in $B_t^\downarrow \setminus B_t$.
    In that case, $u_1$ and $u_\ell$ are connected in $G[B_t^\downarrow \setminus (\bar{S} \cup \{u\})] = G[B_{t'}^\downarrow \setminus \bar{S}]$,
    which implies that $u_1$ and $u_\ell$ belong to the same set $F \in \mathcal{F}'_{t'}$
    (recall that $(S'_{t'},C'_{t'}, D'_{t'}, \mathcal{F}'_{t'})$ is the type of $(t', \bar{S}',\bar{C}')$). 
    Furthermore, because $(t',\bar{S}',\bar{C}')$ has the same type as $(t',S',C')$,
    it follows that $u_1$ and $u_{\ell}$ are connected in $G[B^\downarrow_{t'} \setminus S']$.
    Then, since $u \notin S = S'$ and $u$ is incident to both $u_1$ and $u_\ell$,
    $G[B^\downarrow_t \setminus S]$ contains a cycle.

    To prove that $|\bar{S}| \ge |S|$,
    it suffices to observe that $|\bar{S} \setminus \{u\}| = |\bar{S'}| \ge |S'| = |S \setminus \{u\}|$,
    while $u \in \bar{S}$ if and only if $u \in S$.

    \proofsubparagraph*{Join Nodes.}
    Let $t$ be a Join Node and $t_1$ and $t_2$ be its two children. 
    Notice that since $(t,S,C)$ is a valid partial solution,
    the partial solution $(t_1,S_1,C_1)$ (resp.~$(t_2,S_2,C_2)$) which is the restriction of
    $(t,S,C)$ in $B_{t_1}^\downarrow$ (resp.~$B_{t_2}^\downarrow$) is valid.
    Assume that the claim holds for the nodes $t_1$ and $t_2$.
    In that case, we have stored partial solutions $(t_i,\bar{S}_i,\bar{C}_i)$ of the same type as
    $(t_i,S_i,C_i)$ such that $|\bar{S}_i| \ge |S_i|$, where $i \in \{ 1,2 \}$.
    Furthermore, since $B_t = B_{t_1} = B_{t_2}$ and
    $C_1, C_2$ agree with $C$ in $B_t$,
    we have that $C[v] = C_1[v] = C_2[v] = \bar{C}_1[v] = \bar{C}_2[v]$ for all $v \in B_t$.
    For the same reason it also holds that
    $S \cap B_t = S_1 \cap B_t = S_2 \cap B_t = \bar{S}_1 \cap B_t = \bar{S}_2 \cap B_t$. 

    As it constructs the partial solutions for $t$,
    the algorithm will consider the combination of $(t_1,\bar{S}_1,\bar{C}_1)$ and $(t_2,\bar{S}_2,\bar{C}_2)$.
    Since $\bar{S}_1 \cap B_t = \bar{S}_2 \cap B_t$ and $\bar{C}_1[v]= \bar{C}_2[v]$ for all $v \in B_t$,
    the algorithm stores the partial solution $(t,\bar{S}, \bar{C})$ obtained by said combination
    as long as the graph $G[B_t^{\downarrow} \setminus \bar{S}]$ is acyclic,
    where $\bar{S} = \bar{S}_1 \cup \bar{S}_2$ and $\bar{C} \colon B_t^\downarrow \to [\tw+1]$
    such that $\bar{C}[v] = \bar{C}_1[v]$ if $v \in B_{t_1}^\downarrow$, else $\bar{C}[v] = \bar{C}_2[v]$.
    In the following we argue that (i) the graph $G[B_t^{\downarrow} \setminus \bar{S}]$ is acyclic,
    thus $(t,\bar{S}, \bar{C})$ is a valid partial solution,
    and (ii) the partial solution $(t,\bar{S}, \bar{C})$ has the same type as $(t,S,C)$.

    \begin{claim}\label{claim:tw_acyclicity}
        $G[B_t^{\downarrow} \setminus \bar{S}]$ is acyclic.
    \end{claim}

    \begin{claimproof}
        We start with the following definition of which we will make use in our proof.

        \begin{definition}[Swap]\label{definition:swap}
            Let $w_1 - v_1 - \ldots - v_m - w_2$ be $m+2$ consecutive vertices of a cycle in
            $G[B_t^{\downarrow} \setminus \bar{S}]$.
            We will say that a \emph{swap} is happening between $w_1$ and $w_2$ if:
            \begin{itemize}
                \item $\{v_1, \ldots, v_m\} \subseteq B_t$ and 
                \item one vertex of $\{w_1, w_2\}$ belongs to $B^\downarrow_{t_1} \setminus (\bar{S}_{t_1} \cup B_t)$
                and the other to $B^\downarrow_{t_2} \setminus (\bar{S}_{t_2} \cup B_t)$.
            \end{itemize}
        \end{definition}

        On a high level, given a cycle in $G[B_t^{\downarrow} \setminus \bar{S}]$,
        a swap is happening every time the cycle interchanges between vertices of
        $B^\downarrow_{t_1} \setminus B_t$ and $B^\downarrow_{t_2} \setminus B_t$.
        It is easy to see that $m \ge 1$.

        We are now ready to move on to our proof.
        For the sake of contradiction, assume that $G[B_t^{\downarrow} \setminus \bar{S}]$ contains a cycle,
        and let $\mathcal{C}$ denote one such cycle with a minimum number of swaps,
        where $V(\mathcal{C})$ denotes its vertex set.
        Notice that $\mathcal{C}$ contains vertices from both $B_{t_1}^{\downarrow} \setminus B_t$ and
        $B_{t_2}^{\downarrow} \setminus B_t$,
        as otherwise it exists in either $G[B^\downarrow_{t_1} \setminus \bar{S}]$ or $G[B^\downarrow_{t_2} \setminus \bar{S}]$,
        contradicting the assumption that $(t_1,\bar{S}_1,\bar{C}_1)$ and $(t_2,\bar{S}_2,\bar{C}_2)$ are valid partial solutions.

        Let $\tau$ denote the number of swaps happening between vertices of $\mathcal{C}$;
        since $\mathcal{C}$ contains vertices from both $B^\downarrow_{t_j} \setminus B_t$ where $j \in \{1,2\}$,
        it follows that $\tau \ge 2$.
        For $i \in [\tau]$,
        let $(w^i_1, w^i_2)$ denote the pairs of vertices in $\mathcal{C}$ between of which the swaps are happening,
        indexed by their order of appearance;
        notice that for all $i \in [\tau]$,
        there exists a $j \in \{1,2\}$ such that $w^i_2, w^{(i \bmod \tau)+1}_1 \in B^\downarrow_{t_j} \setminus B_t$,
        that is, both vertices belong to the same subtree rooted at a child $t_j$ of node $t$.
        In that case, it is easy to see that $\tau$ is even.
        Furthermore, let for $i \in [\tau]$,
        $v^i_1 \in B_t$ denote the first vertex in the path between $w^i_1$ and $w^i_2$,
        and $P_i$ be the path between $v^i_1$ and $v^{(i \bmod \tau)+1}_1$, with $V(P_i)$ denoting its vertices.
        In that case, $\mathcal{C}$ can be obtained from $v^1_1-P_1-\ldots -v^\tau_1 - P_\tau$.

        Noticing that the vertices of $V(P_i) \setminus B_t$ belong to the connected component of
        $w^i_2$ in $G[B^\downarrow_{t_j} \setminus (B_t \cup \bar{S}_j)]$,
        where $w^i_2 \in B^\downarrow_{t_j} \setminus (B_t \cup \bar{S}_j)$,
        leads to \cref{obs:tw1}.

        \begin{observation}\label{obs:tw1}
            For all $i \in [\tau]$ it holds that $V(P_i) \setminus B_t \neq \varnothing$.
            Moreover, $V(P_i) \subseteq B^\downarrow_{t_j} \setminus \bar{S}_j$ if and only if
            $w^i_2 \in B^\downarrow_{t_j} \setminus (B_t \cup \bar{S}_j)$.
        \end{observation}

        Define the function $f \colon [\tau] \to \{1,2\}$ which denotes for each path $P_i$
        the subtree defined by the children of $t$ that contains all the vertices of $P_i$,
        that is, for all $i \in [\tau]$,
        $V(P_i) \subseteq B^\downarrow_{t_{f(i)}} \setminus \bar{S}_{f(j)}$.
        Notice that by \cref{definition:swap,obs:tw1},
        it holds that $f(i) \neq f((i \bmod \tau) + 1)$ for all $i \in [\tau]$.

        \begin{observation}\label{obs:tw2}
            For all $i \in [\tau]$,
            there is no path connecting $v^i_1, v^{(i \bmod \tau)+1}_1$ using only vertices of $B_{t}\setminus \bar{S}$.
        \end{observation}

        To see why \cref{obs:tw2} holds, notice that if there was such a path
        connecting $v^i_1$ and $v^{(i \bmod \tau)+1}_1$,
        then $\mathcal{C}$ is not a cycle of a minimum number of swaps.

        We now show that there exists a cycle in $G[B^\downarrow_t \setminus S]$, leading to a contradiction.
        Consider the paths $\bar{P_1} = v^1_1-P_1-v^2_1$ and $\bar{P_2} =v^2_1-P_2-\ldots - v^\tau_1 -P_\tau-v^1_1$. 
        We will show that we can use $\bar{P_1}$ and $\bar{P_2}$ to construct two distinct paths
        $P'_1$ and $P'_2$ between $v^1_1$ and $v^2_1$ in $G[B^\downarrow_t \setminus S]$.
        
        Without loss of generality, let $V(P_1) \subseteq B^\downarrow_{t_1} \setminus \bar{S}_1$,
        that is, $f(1) = 1$.
        In that case, it holds that $f(i) = 1$ for all odd $i \in [\tau]$,
        otherwise, i.e., if $i \in [\tau]$ is even, $f(i) = 2$.

        For each $i \in [\tau]$,
        let $\bar{U}_i$ be the connected component of $G[B^\downarrow_{t_{f(i)}} \setminus \bar{S}_{f(i)}]$
        such that $V(P_i) \subseteq \bar{U}_i$.
        It holds that for $j \in \{1,2\}$, $(t_j,\bar{S}_j,\bar{C}_j)$ has the same type as $(t_j,S_j,C_j)$,
        thus for each $i \in [\tau]$, there exists a connected component $U_i$ of $G[B^\downarrow_{t_{f(i)}} \setminus S_{f(i)}]$
        such that $U_i \cap B_{t_{f(i)}} = U_i \cap B_t = \bar{U}_i \cap B_t = \bar{U}_i \cap B_{t_{f(i)}}$.

        Since $v^i_1-P_i-v^{(i \bmod \tau)+1}_1$ is a path in $G[B^\downarrow_{t_{f(i)}}\setminus S_{f(i)}]$
        and $V(P_i) \subseteq \bar{U}_i$,
        we have that $\{v^i_1,v^{(i \bmod \tau)+1}_1\}\subseteq \bar{U}_i$. 
        Furthermore, since $\{v^i_1,v^{(i \bmod \tau)+1}_1\} \subseteq B_{t}$,
        we have that $\{v^i_1,v^{(i \bmod \tau)+1}_1\} \subseteq U_i$.
        Therefore, for any $i \in [\tau]$,
        there exists a path $Q_i$ connecting $v^i_1$ and $v^{(i \bmod \tau)+1}_1$ in the graph $G[U_i]$.
        Substituting the paths $P_i$ with the paths $Q_i$ for each $i \in [\tau]$,
        results into two walks $v^1_1-Q_1-v^2_1$ and $v^2_1-Q_2-\ldots-v^\tau_1-Q_\tau-v^1_1$ in $G[B^\downarrow_t\setminus S]$.
        Notice that, by construction, $v^1_1-Q_1-v^2_1$ is a path.
        Also, the existence of the walk $v^2_1-Q_2-\ldots-v^\tau_1-Q_\tau-v^1_1$ results in a path $v^2_1-Q-v^1_1$ that uses vertices only
        from the vertex set $\{v^1_1\} \cup \bigcup_{i=2}^\tau (V(Q_i) \cup \{v^i_1\})$.
        It remains to show that the two paths are not identical. 
    
        Notice that $v^1_1-Q_1-v^2_1$ uses vertices from $U_1$.
        Furthermore, by \cref{obs:tw1} we know that $V(Q_1) \setminus B_t \neq \varnothing$. 
        Consider a vertex $x \in V(Q_1) \setminus B_t$; we will prove that $x \notin V(Q)$.
        Assume that $x \in V(Q)$.
        Since $x \notin B_t$ and $V(Q) \setminus B_t \subseteq \bigcup_{i=2}^\tau V(Q_i)$,
        we have that there exists $i \in [2,\tau]$ such that $x \in U_i$.
        Notice that since $x \in U_1 \setminus B_t$,
        we have that $U_i = U_1$.
        Furthermore, since $U_i = U_1 \subseteq B^\downarrow_{t_1}$,
        it holds that $i$ is odd (implying that $i \notin \{2,\tau\}$),
        as well as that $\tau \ge 4$.

        We will use the assumption that $\mathcal{C}$ has a minimum number of swaps to show that this cannot be the case.
        Consider the path $v^2_1-P_2-\ldots-v^{i-1}_1 - P_{i-1} - v^i_1$.
        Notice that, since $i \ge 3$,
        said path includes vertices from $B^\downarrow_{t_2} \setminus B_t$ as it holds that
        $V(P_2) \subseteq B^\downarrow_{t_2}$ and $V(P_2) \not\subseteq B_t$.
        We claim that there exists a path $P$ between $v^1_1$ and $v^i_1$ that uses only vertices of $\bar{U}_1$.
        Recall that $U_1 \cap B_t = \bar{U}_1 \cap B_t$.
        Furthermore, since $\{v^2_1,v^i_1\} \subseteq U_1 \cap B_t$,
        we have that $\{v^2_1,v^i_1\} \subseteq \bar{U}_1$.
        Since $\bar{U}_1$ is a connected component of $G[B^\downarrow_{t_1}\setminus \bar{S}_1]$,
        we have that there exists a path $P$ that connects $v^2_1$ with $v^i_1$ and $V(P) \subseteq \bar{U}_1$.
        In that case, define $\mathcal{C}'$ to be the cycle obtained by $\mathcal{C}$ by substituting
        $v^2_1 - P_2 - \ldots - v^{i-1}_1 - P_{i-1} - v^i_1$
        with $v^2_1 - P - v^i_1$.
        It is easy to see that the number of swaps happening in $\mathcal{C}'$ is less than those happering in $\mathcal{C}$,
        a contradiction.
    \end{claimproof}

    Notice that in the proof of \cref{claim:tw_acyclicity}
    we have shown that any two vertices $v_1,v_2 \in B_t \setminus \bar{S}$
    that are connected in $G[B_t^\downarrow \setminus \bar{S}]$,
    are also connected in $G[B_t^\downarrow \setminus S]$.
    Using analogous arguments one can show that the converse is also true, thus leading to \cref{observation:connectivity_in_join_nodes}.

    \begin{observation}\label{observation:connectivity_in_join_nodes}
        Two vertices $v_1,v_2 \in B_t \setminus S$ are in the same connected component of
        $G[B_t^\downarrow \setminus S]$ if and only if they are in the same connected component of
        $G[B_t^\downarrow \setminus \bar{S}]$.
    \end{observation}

    We now show that $(t,S,C)$ and $(t,\bar{S},\bar{C})$ have the same type.
    Let $(S_t, C_t, D_t, \mathcal{F}_t)$ and
    $(\bar{S}_t, \bar{C}_t, \bar{D}_t, \bar{\mathcal{F}}_t)$
    be the types of $(t,S,C)$ and $(t,\bar{S},\bar{C})$ respectively.
    Due to previous discussion, it holds that $C_t[v] = \bar{C}_t[v]$ for all $v \in B_t$,
    as well as $S_t = \bar{S}_t$.
    Hereinafter, we will use $S_t$ to refer to the set defined by both $S_t$ and $\bar{S}_t$.
    Notice that by \cref{observation:connectivity_in_join_nodes}, it follows that
    both $\mathcal{F}_t$ and $\bar{\mathcal{F}}_t$ define the same partition on the vertices of $B_t \setminus S_t$.

    It remains to prove that $D_t[v] = \bar{D}_t[v]$ for all $v \in S_t$.
    To this end, let $(S_{t_j}, C_{t_j}, D_{t_j}, \mathcal{F}_{t_j})$ denote the types of partial solutions
    $(t_j,S_j,C_j)$ and $(t_j,\bar{S}_j,\bar{C}_2)$, where $j \in \{ 1,2 \}$.
    Observe that
    \begin{itemize}
        \item $(t,S,C)$ is the combination of $(t_1,S_1,C_1)$ and $(t_2,S_2,C_2)$,
        \item and $(t,\bar{S},\bar{C})$ is the combination of $(t_1,\bar{S}_1,\bar{C}_1)$ and $(t_2,\bar{S}_2,\bar{C}_2)$.
    \end{itemize}
    This means that $D_t[v] = \min \{2, D_{t_1}[v] + D_{t_2}[v]\} = \bar{D}_t[v]$ for all $v \in S_t$. 

    Finally, we show that $|\bar{S}| \ge |S|$.
    Recall that $S = S_1 \cup S_2$ and $\bar{S} = \bar{S}_1 \cup \bar{S}_2$. 
    Since $|\bar{S}_1| \ge |S_1|$, $|\bar{S}_2| \ge |{S_2}|$, and $\bar{S}_1 \cap \bar{S}_2 = S_1 \cap S_2$,
    we conclude that $|\bar{S}| = |\bar{S}_1| +|\bar{S}_2| - |\bar{S}_1 \cap \bar{S}_2| \ge |S_1| +|S_2| - |S_1\cap S_2| = |S|$.

    \proofsubparagraph*{Forget Nodes.}
    Let $t$ be a Forget Node, where $t'$ denotes its child node and $u$ the forgotten vertex.
    Assume that the claim holds for node $t'$.
    Note that $(t',S,C)$ is a valid partial solution as well,
    thus there exists a partial solution $(t',\bar{S},\bar{C} )$ stored that has
    the same type as $(t',S,C)$ such that $|\bar{S}| \ge |S|$.
    Let $(S_{t'}, C_{t'}, D_{t'}, \mathcal{F}_{t'})$ be the type of $(t',S,C)$ and $(t',\bar{S},\bar{C})$.
    When the algorithm considers $(t', \bar{S}, \bar{C})$,
    it checks whether $(t,\bar{S},\bar{C})$ is a valid partial solution;
    if not, it discards it. 
    We first prove that $(t,\bar{S},\bar{C})$ is a valid partial solution, and then that it has the same type as $(t,S,C)$.

    We distinguish between two cases, either $u \in S$ or not.
    If $u \in S$, then $u \in \bar{S}$, since $S \cap B_{t'} =\bar{S} \cap B_{t'}$.
    Therefore, the algorithm discards $(t,\bar{S},\bar{C})$ only if $u$ has not found at least two interesting neighbors,
    that is, $D_{t'}[u] < 2$.
    Since $(t',S,C)$ has the same type as $(t',\bar{S},\bar{C})$, it holds that $u$ has not found two interesting neighbors
    in $(t',S,C)$;
    this contradicts to the assumption that $(t,S,C)$ is a valid partial solution. 
    
    If on the other hand $u \notin S$, it follows that $u \notin \bar{S}$.
    We prove that $(t,\bar{S},\bar{C})$ will not be discarded by the algorithm.
    Informally, this will only happen if forgetting $u$ results into two distinct connected components
    that share the same color intersecting the same bag.
    Formally, the algorithm discards $(t,\bar{S},\bar{C})$ if the following holds. 
    Let $\bar{U}$ be the connected component of $G[B_{t'}^\downarrow \setminus \bar{S}]$ such that
    $u \in \bar{U}$, $\bar{U} \cap B_{t'} = \{u\}$,
    and there exists some vertex $v \in B_{t'} \setminus(\bar{S} \cup \{u\}) $ with $\bar{C}[u] = \bar{C}[v]$. 
    Note that since $\bar{U} \cap B_{t'} = \{u\}$, it follows that in the partition defined by
    $\mathcal{F}_{t'}$, vertex $u$ appears in a singleton.
    Now, since $(t',S,C)$ and $(t',\bar{S},\bar{C})$ have the same type,
    and $u$ appears in a singleton set in the partition $\mathcal{F}_{t'}$, 
    we have that $U \cap B_{t'} = \{u\}$, where $U$ is the connected component $U$ of $G[B_{t'}^\downarrow \setminus S]$ where $u \in U$.
    Now, we consider the vertex $v$.
    Recall that $C[v] = \bar{C}[v] = \bar{C}[u] = C[u]$
    and $u,v \in B_{t'}\setminus \bar{S}= B_{t'}\setminus S$.
    Therefore $u$ and $v$ must be in the same connected component of the potential final forest.
    This leads to a contradiction as $u$ and $v$ belong to distinct connected components of
    $G[B_{t'}^\downarrow \setminus S]$, and after node $t$ all the vertices of the connected component of $u$ are forgotten;
    thus $u$ and $v$ will stay disconnected in any potential final forest.

    Consequently, the partial solution $(t,\bar{S},\bar{C})$ is valid and it is stored by the algorithm.
    It  remains to show that $(t,S,C)$ and $ (t,\bar{S},\bar{C})$ have the same type. 
    Let $(S_t, C_t, D_t, \mathcal{F}_t)$ be the type of $(t,S,C)$ and
    $(\bar{S}_t, \bar{C}_t, \bar{D}_t, \bar{\mathcal{F}}_t)$ be the type of $(t,\bar{S},\bar{C})$.
    Since $(t',S,C)$ and $(t',\bar{S},\bar{C})$ have the same type and
    $B_t \subset B_{t'}$,
    it follows that $S_t = \bar{S}_t$ and $C_t[v] = \bar{C}_t[v]$ for all $v \in B_t$.
    Next, observe that forgetting a node does not affect the connectivity in the subtree,
    therefore two vertices of $B_t \setminus S_t$ are connected in $G[B_t^\downarrow \setminus S]$ only if
    that is the case in $G[B_{t'}^\downarrow \setminus S]$ (and analogously for $\bar{S}$).
    Since $(t',S,C)$ and $(t',\bar{S},\bar{C})$ have the same type,
    it follows that $\mathcal{F}_t$ and $\bar{\mathcal{F}}_t$ define the same partition on $B_t \setminus S_t$.

    It remains to prove that $D_t[v] = \bar{D}_t[v]$ for all $v \in S_t$. 
    Note that if $u \in S$,
    then forgetting $u$ does not affect the interesting neighbors of any vertex,
    thus $D_t[v] = D_{t'}[v] = \bar{D}_t[v]$ for all $v \in S_t$.
    Now, we consider the case where $u \notin S$.
    Since $(t',S,C)$ and $(t',\bar{S},\bar{C})$ have the same type,
    $u \notin \bar{S}$ follows.
    Consider a vertex $v \in S_t$. 
    If $u \notin N(v)$ or $C_t[v] \neq C_t[u]$,
    we have that the interesting neighbors of $v$ remain the same as in $B_{t'}$,
    thus $D_t[v] =  D_{t'}[v] =  \bar{D}_t[v]$.
    
    It remains to consider the case where $u \in N(v)$ and $C_t[v] = C_t[u]$;
    notice that then, $D_{t'}[v] \ge 1$.
    If $D_{t'}[v] = 2$, $v$ has already found its two interesting neighbors in both $(t',S,C)$ and $ (t',\bar{S},\bar{C})$,
    therefore, $D_t[v] =  D_{t'}[v] =  \bar{D}_t[v] = 2$.
    It remains to consider the case $D_{t'}[v] = 1$.

    Recall that by Property 1, the same-colored vertices of the partial forest that belong to the same bag
    must be in the same connected component of the potential final forest.
    Therefore, it holds that $D_t[v] = D_{t'}[v] = 1$ if there exists a vertex $v' \in B_t \setminus (S \cup \{u\})$
    such that $C[v'] = C[u]$; if no such vertex exists, $D_t[v] = 0$. 
    Similarly, $\bar{D}_t[v] = D_{t'}[v]$ if there exists a vertex $v'' \in B_t \setminus (\bar{S} \cup \{u\})$
    such that $\bar{C}[v''] = \bar{C}[u]$; otherwise $\bar{D}_t[v] = 0$.
    Based on the discussion so far, it suffices to show that, if $v'$ exists then $v''$ exists and vice versa. 
    That is indeed the case, since
    $\setdef{v' \in B_t \setminus (S \cup \{u\})}{C[v'] = C[u]} = \setdef{v'' \in B_t \setminus (\bar{S} \cup \{u\})}{\bar{C}[v''] = \bar{C}[u]}$,
    as $C$ and $\bar{C}$ agree in the vertices of $B_t$, while $B_t \cap S = B_t \cap \bar{S}$.

    Finally, we need to argue that $|\bar{S}| \ge |S|$.
    This is direct from the assumption that $(t',\bar{S},\bar{C})$ and $(t',{S},{C})$
    have the same type and that we have stored $(t',\bar{S},\bar{C})$ for the node $t'$.
\end{proof}

Finally, we show that we can find a maximum minimal feedback vertex set by considering all partial
solutions of the root node of the tree decomposition.
Let $S^*$ be a maximum minimal feedback vertex set of $G$ and $C^*$ a valid coloring of $V$, with respect to $S^*$.
Let $r$ denote the root node of the tree decomposition.
Notice that $(r, S^*,C^*)$ is a valid partial solution, thus, by \cref{lemma:tw_algo_correctness},
the algorithm stores a partial solution $(r, \bar{S}^*,\bar{C}^*)$
that has the same type as $(r,S^*,C^*)$ and $|\bar{S}^*| \ge |S^*|$.
% Let $(S^*_r,C^*_r,D^*_r,\mathcal{F}^*_r)$ denote this type.
% Note that, since $S^*$ is a maximum minimal feedback vertex set of $G$ and $C^*$ is a valid coloring of $V$
% with respect to $S^*$, it holds that $D^*_r[v] = 2 $ for all $v \in S^* \cap B_r$.
% Additionally, for any $u,v \in B_t \setminus S^*$ with $C^*[v] = C^*[u]$,
% $u$ and $v$ belong to the same partitioning of $\mathcal{F}^*_r$;
% that is the case since $C^*$ is a valid coloring of $V$ with respect to $S^*$,
% therefore any two vertices of the same connected component of $G[V \setminus S^*]$
% have the same color and any pair of distinct connected components that intersect the same bag
% of the tree decomposition have differently-colored vertices. 

\begin{claim}
    $\bar{S}^*$ is a maximum minimal feedback vertex set of $G$.
\end{claim}

\begin{claimproof}
    Since, for any node $t$,
    the algorithm keeps a partial solution only if the considered set is a feedback vertex set of $G[B^\downarrow_t]$,
    it holds that $\bar{S}^*$ is a feedback vertex set of $G[B^\downarrow_r] = G$.
    Furthermore, it holds that $|\bar{S}^*| \geq |S^*|$,
    where $S^*$ is a maximum minimal feedback vertex set of $G$.

    It remains to show that $\bar{S}^*$ is minimal.
    Recall that the bag $B_r$ of the root $r$ of a nice tree decomposition is empty.
    Therefore, for any vertex $u \in \bar{S}^*$ we know that $u \notin B_r$.
    We will prove that $u$ has a private cycle, i.e.,
    that the graph $G[(V \setminus \bar{S}^*)\cup \{u\}]$ is not acyclic. 
    
    Consider the forget node $t$ of the tree decomposition where the forgotten vertex is $u$.
    Let $t'$ be the child node of $t$ and $(\bar{S}_{t'}, \bar{C}_{t'}, \bar{D}_{t'}, \bar{\mathcal{F}}_{t'})$
    be the type of $(t',\bar{S},\bar{C})$,
    where $\bar{S} = \bar{S}^* \cap B^\downarrow_{t'}$ and $\bar{C}$ the restriction of $\bar{C}^*$ in $B^\downarrow_{t'}$.
    Notice that, in order to store $(r,\bar{S}^*,\bar{C}^*)$,
    the algorithm has created and kept the partial solutions $(t',\bar{S},\bar{C})$ and $(t,\bar{S},\bar{C})$.
    Also, since we have kept the partial solution $(t,\bar{S},\bar{C})$ and $u \in \bar{S}$ is the forgotten node of $t$,
    we have that $\bar{D}_{t'}[u]=2$;
    if that were not the case, the algorithm would have discarded $(t,\bar{S},\bar{C})$.
    Let $v_1$ and $v_2$ denote the two interesting neighbors of $u$.
    
    Notice that due to Property 1 it holds that
    \begin{itemize}
        \item either $v_1, v_2$ belong to the same connected component of $G[B_{t}^\downarrow \setminus \bar{S}]$,
        \item or the connected components $U_1, U_2$ of $G[B_{t}^\downarrow \setminus \bar{S}]$,
        where $v_1 \in U_1$ and $v_2 \in U_2$,
        both intersect $B_t$, i.e., $U_i \cap B_t \neq \varnothing$ for $i \in \{1,2\}$.
    \end{itemize}

    In the first case, that is, when $v_1$ and $v_2$ belong to the same connected component of $G[B_t^\downarrow \setminus \bar{S}]$,
    it holds that $G[(B_t^\downarrow \setminus \bar{S}) \cup \{u\}]$ is not acyclic.
    Notice that $G[(B_t^\downarrow \setminus \bar{S}) \cup \{u\}]$ is an induced subgraph of $G[(B_r^\downarrow \setminus \bar{S}^*) \cup \{u\}]$.
    Consequently, $\bar{S}^* \setminus \{u\}$ is not a feedback vertex set of $G$. 
    
    As for the second case,
    let $U_1, U_2$ be the two connected components of $G[B_t^\downarrow \setminus \bar{S}]$,
    where $v_i \in U_i$ and $U_i \cap B_t \neq \varnothing$ for $i \in \{1,2\}$.
    We argue that there exists a connected component $U$ of $G[B_r \setminus \bar{S}^*]$ such that
    $U \supseteq U_1 \cup U_2$.
    Notice that the vertices of $U_1$ and $U_2$ share the same color in $\bar{C}$,
    as $v_1$ and $v_2$ are the interesting neighbors of $u$, thus $\bar{C}[u] = \bar{C}[v_1] = \bar{C}[v_2]$. 
    
    Now, assume that there is no connected component $U$ of $G[B_r\setminus \bar{S}^*]$ such that $U \supseteq U_1 \cup U_2$.
    Also, let $U_1^*$ and $U_2^*$ be the connected components of $G[B_r\setminus \bar{S}^*]$ such that $U_1^*\supseteq U_1$ and $U_2^*\supseteq U_2$. 
    Notice that such connected components exist as $(r,\bar{S}^*, \bar{C}^*)$ extends $(t,\bar{S}, \bar{C})$.
    Furthermore, since $B_t \cap U_i \neq \varnothing$ for $i \in [2]$,
    there exist Forget nodes $t_1$ and $t_2$ which are ancestors of $t$ in the tree decomposition,
    where the last vertex of $U^*_1$ and $U^*_2$ are the forgotten vertices respectively.
    
    Without loss of generality assume that $t_2$ is an ancestor of $t_1$,
    and let $t'_1$ be the child node of $t_1$.
    Let $\bar{S}_1 = B^\downarrow_{t_1} \cap \bar{S}^*$ and $\bar{C}_1$ be the restriction of $\bar{C}^*$ in $B^\downarrow_{t_1}$.
    Notice that, in order to create the partial solution $(r, \bar{S}^*, \bar{C}^*)$,
    the algorithm has created and kept the partial solutions $(t'_1, \bar{S}_1,\bar{C}_1)$ and $(t_1, \bar{S}_1,\bar{C}_1)$. 
    We will show that, in this case, the algorithm discards $(t_1, \bar{S}_1,\bar{C}_1)$ and thus we have a contradiction.

    Notice that, by assumption, $B_{t'_1} \cap U^*_1 = \{v\}$ and $v$ is the forgotten vertex in $t_1$. 
    Furthermore, since $U^*_2 \cap B_t \neq \varnothing$,
    $U^*_2 \cap B_{t_2} \neq \varnothing$,
    and $U^*_2$ is a connected component in $G[B_r \setminus \bar{S}^*]$,
    we can conclude that $U^*_2 \cap B_{t'_1} = U^*_2 \cap B_{t_1} \neq \varnothing$.
    Notice that, in order to create the partial solution $(t_1, \bar{S}_1,\bar{C}_1)$,
    the algorithm considers the partial solution $(t'_1, \bar{S}_1,\bar{C}_1)$.
    Let $v'$ be a vertex in $U^*_2 \cap B_{t'_1}$.
    Since the vertices $v_1$ and $v_2$ have the same color $\bar{C}^*[v_1]$,
    we can conclude that $\bar{C}^*[v] = \bar{C}^*[v_1] = \bar{C}^*[v_2] = \bar{C}^*[v']$.
    Then, the algorithm discards $(t_1, \bar{S}_1,\bar{C}_1)$,
    as $v$ and $v'$ have the same color $\bar{C}_1[v] = \bar{C}^*[v] = \bar{C}^*[v'] = \bar{C}_1[v']$ and $\{v\}$ appears
    as a singleton in the type of $(t'_1, \bar{S}_1,\bar{C}_1)$.
    This contradicts the assumption that the algorithm does not discard $(t_1, \bar{S}_1,\bar{C}_1)$.

    Therefore, there exists a connected component $U$ of $G[B_r\setminus \bar{S}^*]$ such that
    $U \supseteq U_1 \cup U_2$, thus $G[(B_r \setminus \bar{S}^*)\cup \{u\}]$ is not acyclic.
\end{claimproof}

\proofsubparagraph*{Running Time.}
We now consider the running time of our algorithm.
First we provide a bound on the number of valid partial solutions per node of the tree decomposition;
recall that the algorithm stores one such partial solution per type.
Notice that, for any node $t$ of the tree decomposition,
the type of a partial solution $(t,S,C)$ is composed of
(i) a subset $S_t \subseteq B_t$,
(ii) a coloring $C_t \colon B_t \to [\tw+1]$,
(iii) a table $D_t \colon S_t \to \{0,1,2\}$,
and (iv) a partition $\mathcal{F}_t$ of $B_t \setminus S_t$. 
Since $|B_t| \le \tw+1$, it follows that the total number of types for a node $t$ is at most
$2^{\tw+1} \cdot (\tw+1)^{\tw+1} \cdot 3^{\tw+1} \cdot (\tw+1)^{\tw+1}$,
leading to a total of $\tw^{\bO(\tw)}$ different types.

Checking whether a partial solution is valid induces a polynomial-time overhead.
Now, notice that in the dynamic programming part of the algorithm
we can create all the partial solutions (of different types) for Introduce and Forget Nodes
in time $P \cdot |V|^{\bO(1)}$ where $P$ is the number of partial solutions we have stored for
the child of the node we consider.
Therefore, we can compute all partial solutions (of different types)
for these nodes in $\tw^{\bO(\tw)} |V|^{\bO(1)}$ time.
For the Join Nodes, in the worst case we need to consider all pairs of partial solutions corresponding
to the two children of the Join Node.
However, as all the other calculations remain polynomial in the number of vertices,
the time required to compute the partial solutions for Join Node is again $\tw^{\bO(\tw)} |V|^{\bO(1)}$.
In that case, since the tree decomposition consists of a polynomial number of nodes,
it follows that the total running time of the algorithm is $\tw^{\bO(\tw)} n^{\bO(1)}$.
\end{proof}

\section{ETH Lower Bound} \label{sec:eth_vc_lb}

In this section we present a lower bound on the complexity of solving {\mmFVS} parameterized by the vertex cover
of the input graph.  
Starting from a 3-\textsc{SAT} instance on $n$ variables,
we produce an equivalent {\mmFVS} instance on a graph of vertex cover number $\bO (n / \log n)$,
hence any algorithm solving the latter problem in time $\vc^{o(\vc)}n^{\bO(1)}$ would refute the ETH.
As already mentioned, vertex cover is a very restrictive structural parameter,
and due to its known relationship with treewidth,
an analogous lower bounds follows for the latter,
rendering the algorithm of \cref{sec:tw_algo} optimal.
We first state the main theorem of the section.

\begin{theorem}\label{thm:eth_vc_lb}
    There is no $\vc^{o (\vc)}n^{\bO(1)}$ time
    algorithm for \mmFVS, where \vc{} denotes the size of the minimum vertex cover
    of the input graph, unless the ETH fails.
\end{theorem}

Before we present the details of our construction, let us give some high-level
intuition. Our goal is to ``compress'' an $n$-variable instance of
\textsc{3-SAT}, into a {\mmFVS} instance with vertex cover roughly $n / \log n$.
To this end, we will construct $\log n$ choice gadgets, each of which is
supposed to represent $n/\log n$ variables, while contributing only $n/\log^2n$
to the vertex cover. Hence, each vertex of each such gadget must be capable of
representing roughly $\log n$ variables.

Our choice gadget may be thought of as a variation of a bipartite graph with
sets $L,R$, of size roughly $n/\log^2n$ and $\sqrt{n}$ respectively. If one
naively tries to encode information in such a gadget by selecting which
vertices of $L \cup R$ belong in an optimal solution, this would only give $2$
choices per vertex, which is not efficient enough. Instead, we engineer things
in a way that all vertices of $L \cup R$ must belong in the forest in an optimal
solution, and the interesting choice for a vertex $\ell \in L$ is which
vertex $r \in R$ belongs to the same component. In this sense, a
vertex $\ell \in L$ has $|R|$ choices, which is sufficient to encode the
assignment for $\Omega(\log n)$ variables. What remains, then, is to add
machinery that enforces this basic setup, and then clause checking vertices
which for each clause verify that the clause is satisfied.
This is done by testing if a vertex $\ell$ that represents one of its literals is in the same component as
a vertex $r$ that represents a satisfying assignment for the clause.

\subsection{Preliminary Tools}

Before we present the construction that proves \cref{thm:eth_vc_lb}, we give a
variant of \textsc{3-SAT} from which it will be more convenient to start our
reduction, as well as a basic force gadget that we will use in our construction
to ensure that some vertices must be placed in the forest in order to have
a sufficiently large minimal feedback vertex set.

\subparagraph*{\textsc{3-SAT} Variant.}
In the following, we formally define a constrained version of 3-SAT,
called \SD, and establish its hardness under the ETH.
We mention in passing that this variant is an important building block for a recently introduced technique to obtain double-exponential lower bounds for problems in NP parameterized by the treewidth and the vertex cover number~\cite{isaac/ChakrabortyFMT24,icalp/FoucaudGK0IST24}.

\problemdef{\SD}
{A formula $\phi$ in $3$-CNF form, together with a partition
of the set of its variables into three disjoint sets $V_1,V_2,V_3$ with $|V_i| = n$ for $i \in [3]$,
such that no clause contains more than one variable from $V_i$, for all $i \in [3]$.}
{Determine whether $\phi$ is satisfiable.}

\begin{theoremrep}\label{thm:3_sat_dif}
    {\SD} cannot be decided in time $2^{o(n)}$,
    unless the ETH fails.
\end{theoremrep}

\begin{proof}
    Let $\phi$ be a 3-SAT formula of $m$ clauses,
    where $V$ denotes the set of its variables and $|V| = n$.
    We will construct an equivalent instance $\phi'$ of {\SD} as follows:
    \begin{itemize}
        \item For every variable $x \in V$, introduce variables $x_i \in V_i$, for $i \in \bracks{3}$.
        
        \item For every clause $x \lor y \lor z$ of $\phi$,
        introduce a clause $x_1 \lor y_2 \lor z_3$ in $\phi'$.
        In an analogous way, for every clause $x \lor y$ of $\phi$,
        introduce a clause $x_1 \lor y_2$ in $\phi'$.
    
        \item Introduce clauses $\neg x_1 \lor x_2$, $\neg x_2 \lor x_3$, and $\neg x_3 \lor x_1$ in $\phi'$.
        Note that these clauses are all satisfied if and only if variables $x_1$, $x_2$, and $x_3$ share the same assignment.
    \end{itemize}
    Let $V' = V_1 \cup V_2 \cup V_3$. Notice that this is a valid {\SD} instance
    since $|V_i| = n$ for $i \in [3]$, and none of the $m + 3n$ clauses of $\phi'$
    contains more than one variable from $V_i$, for all $i \in [3]$.
    
    We argue that $\phi$ is satisfiable if and only if $\phi'$ is satisfiable.
    For the forward direction, let $f \colon V \to \{ T, F \}$ be an assignment that satisfies $\phi$,
    and notice that the assignment $f' \colon V' \to \{ T, F \}$,
    where $f'(x_i) = f(x)$ for $i \in \bracks{3}$ and $x \in V$,
    satisfies $\phi'$.
    For the converse direction, let $f' \colon V' \to \{ T, F \}$ be an assignment that satisfies $\phi'$,
    and notice that it holds that $f'(x_1) = f'(x_2) = f'(x_3)$.
    Then, it holds that the assignment $f \colon V \to \{ T, F \}$,
    where $f(x) = f(x_1)$ for $x \in V$,
    satisfies $\phi$.
    
    Lastly, assume there exists a $2^{o(|V_i|)}$ algorithm deciding whether $\phi'$ is satisfiable.
    Then, since $|V_i|$ is equal to the number of variables of $\phi$,
    \textsc{3-SAT} could be decided in $2^{o(n)}$, thus the ETH fails.
\end{proof}

\subparagraph*{Force gadgets.}
We now present a gadget that will ensure that a vertex $u$ must be placed in the forest
in any solution that finds a large minimal feedback vertex set.
In the remainder, suppose that $A$ is a sufficiently large value
(we give a concrete value to $A$ in the next section).
When we say that we attach a \emph{force gadget} to a vertex $u$,
we introduce $A+1$ new vertices $\bar{u}, u_1, \ldots, u_A$ to the graph such that
the vertices $u_i$ form an independent set,
while there exist edges $\braces{u, u_i}, \braces{\bar{u},u_i}$ for all $i \in [A]$,
as well as the edge $\braces{u, \bar{u}}$.
See \cref{fig:force_gadget} for an illustration.
We refer to the vertex $\bar{u}$ as the \emph{gadget twin} of $u$,
while the rest of the vertices will be referred to as the \emph{gadget leaves} of $u$.
Intuitively, the idea here is that if $u$
(or $\bar{u}$) is contained in a minimal feedback vertex set, then none of its
$A$ gadget leaves can be also contained, as these vertices cannot have
private cycles. Hence, setting $A$ to be sufficiently large will allow us to
force $u$ to be in the forest.

\begin{figure}[ht]
\centering
\begin{tikzpicture}[scale=0.75, transform shape]

%%%%%%%%%% vertices and text
\node[vertex] (vu) at (0.3,4) {};
\node[vertex] (vv) at (0.3,2) {};

\node[] () at (0,4) {$u$};
\node[] () at (0,2) {$\bar{u}$};

\node[vertex] (vv1) at (2,1) {};
\node[vertex] (vv2) at (2,3) {};
\node[vertex] (vvA) at (2,5) {};

\node[] () at (2.4,1) {$u_1$};
\node[] () at (2.4,3) {$u_2$};
\node[] () at (2,4) {$\vdots$};
\node[] () at (2.4,5) {$u_A$};

%%%%%%%%% edges / arcs
\draw[] (vu)--(vv);

\draw[] (vu)--(vv1);
\draw[] (vu)--(vv2);
\draw[] (vu)--(vvA);

\draw[] (vv)--(vv1);
\draw[] (vv)--(vv2);
\draw[] (vv)--(vvA);

\end{tikzpicture}
\caption{Force gadget attached to vertex $u$.}
\label{fig:force_gadget}
\end{figure}

\subsection{Construction}

Let $\phi$ be a {\SD} instance of $m$ clauses, where $|V_p| = n$ for $p \in [3]$ and,
without loss of generality, assume that $n$ is a power of $4$
(this can be achieved by adding dummy variables to the instance if needed).
Partition each variable set $V_p$ into $\log n$ subsets $V_p^q$ of size at most
$\ceil{\frac{n}{\log n}}$, where $p \in [3]$ and $q \in [\log n]$.
Let $L = \ceil{\frac{n}{\log^2 n}}$.
Moreover, partition each variable subset $V_p^q$ into $2L$ subsets
$\mathcal{V}^{p,q}_\alpha$ of size as equal as possible, where $\alpha \in [2L]$.
Define $R = \sqrt{n}$, $A = n^2 + m$, and $k = (4AL + AR + 2LR) \cdot 3\log n + m$.
We will proceed with the construction of a graph $G$ such that
$\phi$ is satisfiable if and only if $G$ has a minimal feedback vertex set of size at least $k$.

For each variable subset $V_p^q$ we define the \emph{choice gadget graph $G_p^q$} as follows
(see also \cref{fig:choice_gadget_b}):
\begin{itemize}
    \item $V(G_p^q) = \setdef{\ell_\alpha, \ell'_\alpha, \kappa_\alpha, \lambda_\alpha}{\alpha \in [2L]} \cup \setdef{r_\beta}{\beta\in[R]} \cup \setdef{m_\beta^\alpha}{\alpha \in [2L], \beta \in [R]}$,
    \item all the vertices $\ell_\alpha$, $\ell'_\alpha$, and $r_\beta$ have an attached force gadget,
    \item for $\alpha \in [2L]$, $N(\kappa_\alpha) = M_\alpha \cup \braces{\lambda_\alpha}$
    and $N(\lambda_\alpha) = M_\alpha \cup \braces{\kappa_\alpha}$,
    where $M_\alpha = \setdef{m_\beta^\alpha}{\beta \in [R]}$,
    \item for $\alpha \in [2L]$ and $\beta \in [R]$,
    $m_\beta^\alpha$ has an edge with $\ell_\alpha$, $\ell'_\alpha$, and $r_\beta$.
\end{itemize}
We additionally define the
\emph{choice set $X_\alpha$} as $X_\alpha = M_\alpha \cup \braces{\kappa_\alpha, \lambda_\alpha}$.

Intuitively, one can think of this gadget as having been constructed as
follows: we start with a complete bipartite graph that has on one side the
vertices $\ell_\alpha$ and on the other the vertices $r_\beta$; we subdivide each edge
of this graph, giving the vertices $m_\beta^\alpha$; for each $\alpha \in [2L]$ we add
$\ell'_\alpha,\kappa_\alpha,\lambda_\alpha$, connect them to the same $m_\beta^\alpha$ vertices that
$\ell_\alpha$ is connected to and connect $\kappa_\alpha$ to $\lambda_\alpha$; we attach force
gadgets to all $\ell_\alpha, \ell'_\alpha, r_\beta$. Hence, as sketched before, the idea of
this gadget is that the choice of a vertex $\ell_\alpha$ is to pick an $r_\beta$ with
which it will be in the same component in the forest, and this will be
expressed by picking one $m_\beta^\alpha$ that will be placed in the forest.

\begin{figure}[ht]
\centering 
  \begin{subfigure}[b]{0.45\linewidth}
  \centering
    \begin{tikzpicture}[scale=0.75, transform shape]
    %%%%%%%%%% vertices and text
    \node[black_vertex] (vl) at (0.3,3.5) {};
    \node[black_vertex] (vl') at (0.3,1.5) {};
    
    \node[] () at (0,3.5) {$\ell_\alpha$};
    \node[] () at (0,1.5) {$\ell'_\alpha$};
    
    \node[vertex] (vm1) at (3,0.5) {};
    \node[vertex] (vm2) at (3,2.5) {};
    \node[vertex] (vmsn) at (3,4.5) {};
    
    \node[] () at (3.4,0.2) {$m_1^\alpha$};
    \node[] () at (3.4,2.2) {$m_2^\alpha$};
    \node[] () at (3,3.5) {$\vdots$};
    \node[] () at (3.4,4.2) {$m_R^\alpha$};
    
    \node[black_vertex] (vr1) at (6,0.5) {};
    \node[black_vertex] (vr2) at (6,2.5) {};
    \node[black_vertex] (vrsn) at (6,4.5) {};
    
    \node[] () at (6,0.2) {$r_1$};
    \node[] () at (6,2.2) {$r_2$};
    \node[] () at (6,3.5) {$\vdots$};
    \node[] () at (6,4.2) {$r_R$};
    
    \node[vertex] (vkappa) at (2,5.5) {};
    \node[vertex] (vlambda) at (4,5.5) {};
    
    \node[] () at (2,5.8) {$\kappa_\alpha$};
    \node[] () at (4,5.8) {$\lambda_\alpha$};
    
    %%%%%%%%% edges / arcs
    \draw[] (vkappa)--(vlambda);
    
    \draw[] (vl)--(vm1);
    \draw[] (vl)--(vm2);
    \draw[] (vl)--(vmsn);
    
    \draw[] (vl')--(vm1);
    \draw[] (vl')--(vm2);
    \draw[] (vl')--(vmsn);
    
    \draw[] (vr1)--(vm1);
    \draw[] (vr2)--(vm2);
    \draw[] (vrsn)--(vmsn);
    
    \draw[] (vkappa) edge [bend right] (vm1);
    \draw[] (vkappa) edge [bend right] (vm2);
    \draw[] (vkappa)--(vmsn);
    
    \draw[] (vlambda) edge [bend left] (vm1);
    \draw[] (vlambda) edge [bend left] (vm2);
    \draw[] (vlambda)--(vmsn);
    
    \end{tikzpicture}
    \caption{Part of the construction concerning $X_\alpha$.}
    \label{fig:choice_gadget_a}
  \end{subfigure}
\begin{subfigure}[b]{0.45\linewidth}
\centering
  \begin{tikzpicture}[scale=0.75, transform shape]
    %%%%%%%%%% vertices and text
    \node[black_vertex] (vl) at (0.3,6.5) {};
    \node[black_vertex] (vl') at (0.3,5.5) {};
    
    \node[] () at (0,6.4) {$\ell_{2L}$};
    \node[] () at (0,5.4) {$\ell'_{2L}$};

    \node[] () at (0.3,4) {$\vdots$};

    \node[black_vertex] (vl2) at (0.3,2.5) {};
    \node[black_vertex] (vl2') at (0.3,1.5) {};
    
    \node[] () at (0,2.5) {$\ell_1$};
    \node[] () at (0,1.5) {$\ell'_1$};

    \node[vertex] (vm1) at (3,5.5) {};
    \node[vertex] (vmsn) at (3,6.5) {};
    
    \node[] () at (3,5.1) {$m_1^{2L}$};
    \node[] () at (3,6.1) {$\vdots$};
    \node[] () at (3.5,6.7) {$m_R^{2L}$};

    \node[vertex] (vm21) at (3,1.5) {};
    \node[vertex] (vms2n) at (3,2.5) {};
    
    \node[] () at (3.4,1.3) {$m_1^1$};
    \node[] () at (3,2.1) {$\vdots$};
    \node[] () at (3.4,2.4) {$m_R^1$};

    \node[vertex] (vkappa) at (2,8) {};
    \node[vertex] (vlambda) at (4,8) {};
    
    \node[] () at (2,8.3) {$\kappa_{2L}$};
    \node[] () at (4,8.3) {$\lambda_{2L}$};

    \node[vertex] (vkappa2) at (2,4) {};
    \node[vertex] (vlambda2) at (4,4) {};
    
    \node[] () at (2,4.3) {$\kappa_1$};
    \node[] () at (4,4.3) {$\lambda_1$};

    \node[black_vertex] (vr1) at (6,2) {};
    \node[black_vertex] (vrsn) at (6,6) {};
    
    \node[] () at (6,1.6) {$r_1$};
    \node[] () at (6,4) {$\vdots$};
    \node[] () at (6,6.4) {$r_R$};

    %%%%%%%%% edges / arcs
    \draw[] (vkappa)--(vlambda);
    
    \draw[] (vl)--(vm1);
    \draw[] (vl)--(vmsn);
    
    \draw[] (vl')--(vm1);
    \draw[] (vl')--(vmsn);
    
    \draw[] (vr1) edge [bend right] (vm1);
    \draw[] (vrsn)--(vmsn);
    
    \draw[] (vkappa) edge [bend right] (vm1);
    \draw[] (vkappa)--(vmsn);
    
    \draw[] (vlambda) edge [bend left] (vm1);
    \draw[] (vlambda)--(vmsn);

    \draw[] (vkappa2)--(vlambda2);
    
    \draw[] (vl2)--(vm21);
    \draw[] (vl2)--(vms2n);
    
    \draw[] (vl2')--(vm21);
    \draw[] (vl2')--(vms2n);
    
    \draw[] (vr1)--(vm21);
    \draw[] (vrsn) edge [bend left] (vms2n);
    
    \draw[] (vkappa2) edge [bend right] (vm21);
    \draw[] (vkappa2)--(vms2n);
    
    \draw[] (vlambda2) edge [bend left] (vm21);
    \draw[] (vlambda2)--(vms2n);
    
    \end{tikzpicture}
    \caption{The choice gadget graph $G^q_p$.}
    \label{fig:choice_gadget_b}
  \end{subfigure}
\caption{Black vertices have a force gadget attached.}
\label{fig:choice_gadget}
\end{figure}  

Notice that each vertex $\ell_\alpha$ of $G_p^q$, where $\alpha \in [2L]$, $p \in [3]$, and $q \in [\log n]$,
is used to represent a variable subset
$\mathcal{V}^{p,q}_\alpha \subseteq V_p^q$ containing at most 
\[
    |\mathcal{V}^{p,q}_\alpha| \leq
    \ceil*{\frac{|V^{p,q}|}{2L}} \leq
    \ceil*{\frac{\ceil{\frac{n}{\log n}}}{2L}} =
    \ceil*{\frac{n}{2L \log n}} \leq
    \ceil*{\frac{n}{2 \frac{n}{\log^2 n} \log n}} =
    \ceil*{\frac{\log n}{2}} =
    \frac{\log n}{2}
\]
variables of $\phi$, where we used \cref{thm:ceilings} for $f(x) = x / 2L$.
We fix an arbitrary one-to-one mapping so that every
vertex $m_\beta^\alpha$, where $\alpha \in [2L]$ and $\beta \in [R]$, corresponds to a different
assignment for this subset, which is dictated by which element of $M_\alpha$
was not included in the final feedback vertex set.  Since $R = 2^{\log n / 2}
= \sqrt{n}$, the size of $M_\alpha$ is sufficient to uniquely encode all the
different assignments of $\mathcal{V}^{p,q}_\alpha$.

Finally, introduce \emph{clause vertices} $c_1, \ldots, c_m$, each of which corresponds to a clause of $\phi$,
and define the graph $G$ as the union of these vertices as well as all graphs $G_p^q$, where $p \in [3]$ and $q \in [\log n]$.
For a clause vertex $c$, add an edge to $\ell_\alpha$ in $G_p^q$ when $\mathcal{V}^{p,q}_\alpha$
contains a variable appearing in $c$,
as well as to the vertices $r_\beta$ for each such $\ell_\alpha$,
such that $m^\alpha_\beta \notin S$ corresponds to an assignment of $\mathcal{V}^{p,q}_\alpha$ satisfying $c$,
where $S$ denotes a minimal feedback vertex set.
Notice that since no clause contains multiple variables from the same variable set $V_i$,
due to the refinement of the partition of the sets of variables, it holds that all
the variables of a clause will be represented by vertices appearing in distinct choice gadget graphs $G_p^q$.

\subsection{Correctness}
Having constructed the previously described instance $(G, k)$ of \mmFVS,
it remains to prove its equivalence with the initial {\SD} instance.

\begin{lemmarep}\label{lem:size_per_subgraph}
    Any minimal feedback vertex set $S$ of $G$ of size at least $k$ has the following properties:
    \begin{romanenumerate}
        \item $S$ does not contain any vertex attached with a force gadget or its gadget twin,
        \item $|M_\alpha \setminus S| \leq 1$, for every $G_p^q$ and $\alpha \in [2L]$,
        \item $|S \cap V(G_p^q)| = 4AL + AR + 2LR$,
    \end{romanenumerate}
    where $p \in [3]$ and $q \in [\log n]$.
\end{lemmarep}

\begin{proof}
    Let $S$ be a minimal feedback vertex set of size $|S| \geq k > (4L + R) \cdot 3A \log n$.
    Let $u$ be a vertex attached with a force gadget, and $\bar{u}$ its gadget twin.
    
    For the first property, suppose that $u, \bar{u} \in S$.
    In that case, $S \setminus \braces{\bar{u}}$ remains a feedback vertex set of $G$, thus $S$ cannot be minimal.
    On the other hand, if one of $u, \bar{u}$ belongs to $S$, then $|S| \leq |V(G)| - (A + 1)$,
    since $S$ cannot include the rest of the vertices of the corresponding force gadget,
    due to minimality (notice that the gadget leaves in that case cannot have a private cycle).
    However, for the defined $A$ and sufficiently large $n$, this leads to a contradiction, since
    \begin{align*}
        (4L + R) \cdot 3A \log n \leq |V(G)| - A - 1 \iff\\
        (4L + R) \cdot 3A \log n \leq m + (8L + 4AL + 2R + AR + 2L (2 + R)) 3 \log n - A - 1 \iff\\
        n^2 \leq (12L + 2R + 2LR) 3 \log n - 1 = \bO \parens*{\frac{n \sqrt{n}}{\log n}}.
    \end{align*}
    Consequently, it follows that $u, \bar{u} \notin S$ for any vertex $u$ attached with a force gadget.    

    For the second property,
    consider $G_p^q$ for some $p \in [3]$ and $q \in [\log n]$,
    and $Y_\alpha = S \cap X_\alpha$ for choice set $X_\alpha$,
    where $\alpha \in [2L]$.
    Since $S$ does not contain any vertices attached with a force gadget,
    it must contain at least $R-1$ vertices of $M_\alpha$;
    if not, there exists a cycle involving vertices $\ell_\alpha, m_\beta^\alpha, \ell'_\alpha, m_{\beta'}^\alpha$
    for some $\beta,\beta' \in [R]$.
    Therefore, $|M_\alpha \setminus S| \leq 1$.

    Lastly, observe that if $|M_\alpha \setminus S| = 0$, then $|Y_\alpha| \geq R$.
    If that is not the case, i.e., $|M_\alpha \setminus S| = 1$, $S$ must contain an additional vertex of $X_\alpha$,
    since a cycle involving vertices $\kappa_\alpha, \lambda_\alpha, m_\beta^\alpha$ remains otherwise.
    Hence, in both cases, it follows that $|Y_\alpha| \geq R$.
    Suppose that $|Y_\alpha| > R$.
    In that case, if $M_\alpha \subseteq Y_\alpha$, then $Y_\alpha$ contains at least one of $\kappa_\alpha$ and $\lambda_\alpha$.
    However, $S' = S \setminus \{ \kappa_\alpha, \lambda_\alpha \}$ remains a feedback vertex set,
    thus $S$ is not minimal.
    Alternatively, $Y_\alpha$ contains both $\kappa_\alpha, \lambda_\alpha$ and all but one element of $M_\alpha$.
    However, $S' = S \setminus \{ \lambda_\alpha \}$ remains a feedback vertex set, thus $S$ is not minimal.
    Since $S$ includes $A$ vertices per force gadget and exactly $R$ vertices per choice set,
    the last property follows.
\end{proof}

\begin{lemmarep}\label{lem:vc_lb_correctness1}
    If $\phi$ has a satisfying assignment,
    then $G$ has a minimal feedback vertex set of size at least $k$.
\end{lemmarep}

\begin{proof}
    Assume that $\phi$ has a satisfying assignment $f \colon V \to \{ T, F \}$.
    For each set of variables $V_p^q$, where $p \in [3]$ and $q \in [\log n]$,
    consider the corresponding choice gadget graph $G_p^q$.
    For all $\alpha \in [2L]$,
    vertex $\ell_\alpha$ in $G_p^q$ represents a subset $\mathcal{V}^{p,q}_\alpha \subseteq V_p^q$,
    for which there exists a $\beta \in [R]$ such that $m^\alpha_\beta$ corresponds
    to the restriction of $f$ to $\mathcal{V}^{p,q}_\alpha$.
    Moreover, each variable $x \in V_p^q$ is uniquely represented by some vertex $\ell_\alpha$ in $G_p^q$.
    Let $S$ be a set of size $k$ containing
    \begin{itemize}
        \item all the $A$ gadget leaves per force gadget,

	\item all the $2L \cdot 3 \log n$ vertices $\kappa_\alpha$,
        
        \item $m_{\beta'}^\alpha$, with $\beta' \neq \beta$,
        for each $G_p^q$ and each subset $\mathcal{V}^{p,q}_\alpha \subseteq V_p^q$,
        where $m^\alpha_\beta$ corresponds to the restriction of $f$ to $\mathcal{V}^{p,q}_\alpha$,

        \item all clause vertices $c_1, \ldots, c_m$.
    \end{itemize}

    We will argue that $S$ is a feedback vertex set.
    Since $S$ contains all the clause vertices,
    the only possible remaining cycles concern vertices in the same choice gadget graph $G_p^q$.
    Since $S$ contains all the gadget leaves per force gadget,
    all the vertices attached with a force gadget do not belong to $S$.
    All $\lambda$ vertices have a single neighbor, hence cannot be part of any cycle.
    Moreover, $\ell$ and $\ell'$ vertices cannot be part of a cycle,
    since they are of degree $2$ and one of their neighbors (their gadget twin) is a leaf.
    Therefore, any possible cycle contains only $r$ and $m$ vertices.
    However, notice that the graph induced by $r$ and $m$ vertices is a union of stars,
    with $r$ vertices being the centers and $m$ vertices being the leaves.
    Consequently, $G - S$ cannot have any cycles.

    We now argue that $S$ is a minimal feedback vertex set.
    Assume there exists $u \in S$ such that $S \setminus \braces{u}$ is a feedback vertex set.
    In that case, $u$ cannot be a gadget leaf introduced by a force gadget,
    since both the vertex it is attached to as well as the latter's gadget twin do not belong to $S$.
    On the other hand, if $u$ were a vertex $m^\alpha_{\beta'}$, then a $\ell_\alpha - m_\beta^\alpha - \ell'_\alpha - m_{\beta'}^\alpha$ cycle would remain.
    Furthermore, if it were a $\kappa_\alpha$ vertex, then a $\kappa_\alpha - \lambda_\alpha - m^\alpha_\beta$ cycle would remain.
    Lastly, $u$ cannot be any clause vertex $c$.
    Indeed, for any $c$, there exists a variable $x$ due to which $c$ is satisfied.
    Consequently, there exists $\ell_\alpha$ in choice gadget graph $G_p^q$ representing $\mathcal{V}^{p,q}_\alpha \ni x$,
    as well as $m^\alpha_\beta \notin S$ encoding said satisfying assignment.
    Therefore, $\ell_\alpha - m_\beta^\alpha - r_\beta - c$ comprises a cycle,
    because we connect $c$ to all $r$ vertices that encode a satisfying
    assignment for $c$.
\end{proof}

\begin{lemmarep}\label{lem:vc_lb_correctness2}
    If $G$ has a minimal feedback vertex set of size at least $k$,
    then $\phi$ has a satisfying assignment.
\end{lemmarep}

\begin{proof}
    Let $S$ denote said minimal feedback vertex set.
    Notice that due to \cref{lem:size_per_subgraph},
    it holds that $|S \cap V(G_p^q)| = 4AL + AR + 2LR$
    for all $p \in [3]$ and $q \in [\log n]$.
    In that case, since $|S| \geq k = (4AL + AR + 2LR) \cdot 3 \log n + m$ and graph $G$
    is composed of all choice gadget graphs $G_p^q$ as well as vertices $c_1, \ldots, c_m$,
    it follows that $c_i \in S$ for all $i \in [m]$.

    Since $S$ is minimal, it holds that,
    for all clause vertices $c$, $S \setminus \{ c \}$ is not a feedback vertex set.
    Consequently, $G - (S \setminus \braces{c})$ contains at least one cycle involving vertex $c$.
    Notice that each such cycle can only involve vertices belonging to a specific choice gadget graph $G_p^q$,
    since vertices not belonging to the same $G_p^q$ can only be connected via paths containing vertices $c_i$,
    but only a single such vertex remains in $G - (S \setminus \braces{c})$.
    Let $G_c = G[(V(G_p^q) \setminus S) \cup \braces{c}]$ be a subgraph of $G$ containing one such cycle.

    We will show that the aforementioned cycle must be of the form $\ell_\alpha - m_\beta^\alpha - r_\beta - c$,
    for some $\alpha \in [2L]$ and $\beta \in [R]$.
    In order to do so, we first argue that there is no path in $G_c - \{c\}$ between any two $r$ vertices.
    Suppose there exists such a path,
    connecting $r_i$ and $r_j$, for distinct $i, j \in [R]$.
    First, notice that any $\kappa$ and $\lambda$ vertices in $G_c$ are leaves due to the second property in \cref{lem:size_per_subgraph},
    thus they cannot be present in said path.
    Furthermore, this path cannot involve only $r$ and $m$ vertices,
    since the graph induced by those is a union of stars,
    with $r$ vertices being the centers and $m$ vertices being the leaves of the stars.
    Therefore, any path from $r_i$ to $r_j$ must include a vertex
    $\ell_\alpha$ or $\ell'_\alpha$ for some $\alpha \in [2L]$,
    denoted by $w_\alpha$.
    In that case, the shortest such path must be of the form
    $r_i - m_i^\alpha - w_\alpha - m_j^\alpha - r_j$.
    However, this cannot be the case, since $G_c$ contains at most one vertex belonging to $M_\alpha$,
    due to \cref{lem:size_per_subgraph}. 
    
    Consequently, any cycle that contains $c$ in $G_c$ must include the unique vertex
    $\ell_\alpha$ that is a neighbor of $c$.
    Moreover, as the only other vertices that are adjacent to $c$ are $r$ vertices,
    and there are no paths between any two $r$ vertices,
    the cycle must be of the form $\ell_\alpha - m_\beta^\alpha - r_\beta - c$ for some $\beta \in [R]$.

    Now, consider the following assignment for the variables of $\phi$:
    for a set of variables $\mathcal{V}^{p,q}_\alpha \subseteq V_p^q$ represented by $\ell_\alpha$ in $G_p^q$,
    if there exists a vertex $m_\beta^\alpha \notin S$ for some $\beta \in [R]$,
    then let these variables have the assignment encoded by this choice.
    Alternatively, if there is no such vertex $m$,
    then set all of these variables to true.
    This is a valid assignment, since every variable of $\phi$ appears in a single variable set $\mathcal{V}^{p,q}_\alpha \subseteq V_p^q$,
    for some $\alpha \in [2L]$, $p \in [3]$, and $q \in [\log n]$,
    which is uniquely represented by a single vertex $\ell_\alpha$ in $G_p^q$,
    while $|M_\alpha \setminus S| \leq 1$. 
    Lastly, this is a satisfying assignment, since (by the minimality of $S$) for every clause vertex $c$,
    there exist neighboring vertices $\ell_\alpha$ and $r_\beta$ such that $m_\beta^\alpha \notin S$,
    i.e., for every clause, there exists at least one variable in $\mathcal{V}^{p,q}_\alpha$
    such that its assignment satisfies said clause.
\end{proof}

\begin{lemmarep}\label{lem:vc}
    It holds that $\vc (G) = \bO (n / \log n)$.
\end{lemmarep}

\begin{proof}
    Notice that the graph resulting from deleting all vertices
    $\ell_\alpha, \ell'_\alpha, r_\beta, \kappa_\alpha, \lambda_\alpha$
    and their gadget twins,
    where $\alpha \in [2L]$ and $\beta \in [R]$,
    from all choice gadget graphs, is an independent set.
    Therefore,
    \[
        \vc(G) \leq (8L + 2R + 4L) \cdot 3\log n = \bO (n / \log n),
    \]
    and the statement follows.
\end{proof}

We now prove \cref{thm:eth_vc_lb}.

\begin{proof}[Proof of \cref{thm:eth_vc_lb}]
    Let $\phi$ be a {\SD} formula.
    In polynomial time, we can construct a graph $G$ such that,
    due to \cref{lem:vc_lb_correctness1,lem:vc_lb_correctness2},
    deciding if $G$ has a minimal feedback vertex set of size at least $k$ is equivalent to
    deciding if $\phi$ has a satisfying assignment.
    In that case, assuming there exists a $\vc^{o(\vc)} n^{\bO(1)}$ algorithm for \mmFVS,
    one could decide {\SD} in time
    \[
        \vc^{o(\vc)} n^{\bO(1)} =
        \parens*{\frac{n}{\log n}}^{o (n / \log n)} n^{\bO(1)} =
        2^{(\log n - \log \log n) o (n / \log n) + \bO(\log n)} =
        2^{o (n)},
    \]
    which contradicts the ETH due to \cref{thm:3_sat_dif}.
\end{proof}

Since for any graph $G$ it holds that $\tw(G) \leq \vc(G) + 1$,
the following corollary holds.

\begin{corollary}\label{cor:eth_tw_lb}
    There is no $\tw^{o (\tw)}n^{\bO(1)}$ time algorithm for \mmFVS,
    where {\tw} denotes the treewidth of the input graph, unless the ETH fails.
\end{corollary}

\section{Natural Parameter Algorithm}\label{sec:natural}

In this section we present an FPT algorithm for {\mmFVS} parameterized by the natural parameter,
i.e., the size of the sought minimal feedback vertex set $k$.
The main theorem of this section is the following.

\begin{theorem}\label{thm:mmfvs_natural_algo}
    {\mmFVS} can be solved in time $9.34^{k} n^{\bO(1)}$.
\end{theorem}

\subparagraph*{Structure of the Section.}
In \cref{subsec:path_restricted} we define the closely related {\ammFVS} problem
and prove that it remains NP-hard, even on some instances of a specific form,
called \emph{path-restricted}.
Subsequently, we present an algorithm dealing with these kinds of instances,
which either returns a minimal feedback vertex set of size at least $k$ or
concludes that this is a No instance of \ammFVS.
Afterwards, in \cref{subsec:general_algorithm}, we solve {\mmFVS}
by producing a number of instances of {\ammFVS} and utilizing the previous algorithm,
thereby proving \cref{thm:mmfvs_natural_algo}.

\subparagraph*{Oversight of \cite{arxiv/GaikwadKMST22}.}
The algorithm of Gaikwad et al.~\cite{arxiv/GaikwadKMST22} performs a branching procedure which
marks vertices as either belonging in the feedback vertex set or the remaining forest.
The flaw is that the algorithm ceases the branching once $k$ vertices have been
identified as vertices of the feedback vertex set.
However, this is not correct, since deciding if a given set $S$ can be extended into
a minimal feedback vertex set $S^* \supseteq S$ is NP-complete and even W[1]-hard
parameterized by $|S|$~\cite{CaselFGMS22}.
Hence, identifying $k$ vertices of the solution is not, in general, sufficient to produce
a feasible solution and the algorithm of~\cite{arxiv/GaikwadKMST22} is incomplete,
as it does not explain how the guessed part of the feedback vertex set can be
extended into a feasible minimal solution.
Intuitively, the pitfall here is that, unlike other standard maximization problems,
such as \textsc{Max Clique}, {\mmFVS} is not monotone, that is, a graph that contains
a feasible solution of size $k$ is not guaranteed to contain a feasible solution of size $k-1$
(consider, for instance, a $K_{2,n}$).

\subsection{Annotated MMFVS and Path-restricted Instances}\label{subsec:path_restricted}

First, we define the following closely related problem, called {\ammFVS} for short.

\problemdef{\textsc{Annotated Maximum Minimal Feedback Vertex Set}}
{A graph $G = (V, E)$, disjoint sets $S, F \subseteq V$ where $S \cup F$ is a feedback vertex set of $G$, as well as an integer $k$.}
{Determine whether there exists a minimal feedback vertex set $S'$ of $G$ of size $|S'| \geq k$ such that $S' \supseteq S$ and $S' \cap F = \varnothing$.}

Notice that if $F$ is not a forest, then the corresponding instance always has a negative answer.
Moreover, given an instance $\mathcal{I}$, let $\ammfvs(\mathcal{I})$ be equal to $1$ if it is a Yes instance and $0$ otherwise.

Before we proceed, let us give some general remarks to explain how the results of this section
fit in with the larger algorithm for the parameter $k$.
We defined an annotated version of our problem because our plan is to use a branching strategy,
marking vertices of the graph as belonging to the minimal feedback vertex set $S$ or
the remaining forest $F$.
We will use $U = V(G) \setminus (S \cup F)$ to denote the remaining (undecided) part of the graph. 

Our main focus in this section is to deal with a (very) restricted special case of our problem:
the case where every component of $G[U]$ is a path, and more strongly the case where for all vertices
$u \in U$ it holds that $\deg_{F \cup U} (u) = 2$.
Note that, since $S\cup F$ is a feedback vertex set, when $\deg_{F \cup U} (u) = 2$ for
the vertices of a component of $G[U]$ this implies that the component is a path whose
internal vertices have no neighbors in $F$.
We call instances that satisfy the above conditions \emph{path-restricted}.

The reason we are interested in path-restricted instances is that the branching algorithm
we present in \cref{subsec:general_algorithm} will consider such instances as a base case
and cease branching once the current instance is path-restricted.
We therefore need an explicit algorithm to deal with such instances.
An astute reader may be wondering whether such instances can in fact be solved in polynomial
time -- after all we have severely restricted the structure of the undecided part $G[U]$.
Alas, this is not the case and we show in \cref{thm:ammfvs_hardness} that such instances
remain NP-complete in a strong sense:
it is hard to distinguish Yes instances from instances where no solution of the desired size
exists even if we permit ourselves to disregard the given annotation.
Given this hardness, we are therefore obliged to give an FPT algorithm for this case (\cref{thm:path_restricted_algo}).

\begin{theoremrep}\label{thm:ammfvs_hardness}
    The following problem is NP-complete:
    given a path-restricted instance $\mathcal{I} = (G,S,F,k)$ of \ammFVS,
    where all components of $G[V\setminus(S\cup F)]$ are paths on $3$ vertices,
    distinguish between the following two cases:
    (i) the given instance is a Yes instance, that is, there exists a minimal feedback vertex set $S'$ of $G$
    of size $|S'| \geq k$ such that $S' \supseteq S$ and $S' \cap F = \varnothing$,
    (ii) any minimal feedback vertex set $S'$ of $G$ has size $|S'| < k$.
\end{theoremrep}

\begin{proof}
    Let the graph $G=(V,E)$, where $|V| = n$ and $|E| = m$, be an instance of \ThCol.
    We will construct an instance $(G', S, F, k)$ of \ammFVS.
    Set $A = 2n + 6$ and construct the graph $G' = (V', E')$ as follows
    \begin{itemize}
        \item introduce the vertex $w \in V'$,

	\item introduce three vertices $w_1, w_2, w_3$ and edges to form the
        cycle $w - w_1 - w_2 - w_3 - w$;
        similarly introduce three vertices $w_1', w_2', w_3'$ and edges to form the
        cycle $w - w_1' - w_2' - w_3' - w$,
        
        \item for every vertex $u_i \in V$, introduce vertices $u_i^{j_1}$ and $v_i^{j_2}$ in $V'$,
        as well as edges $\braces{u_i^{j_1}, v_i^{j_2}}$ in $E'$, where $j_1 \in [3]$ and $j_2 \in [A]$,
        
        \item for every edge $e_i \in E$, introduce vertices $e_i^j \in V'$ and edges $\braces{e_i^j, w} \in E'$, where $j \in [3]$,
        
        \item for all $i \in [n]$, introduce edges to form the
        cycle $w - u_i^1 - u_i^2 - u_i^3 - w$,
        
        \item for every edge $e_i = \braces{u_k, u_\ell} \in E$, introduce edges $\braces{e_i^j, u_k^j}, \braces{e_i^j, u_\ell^j} \in E'$, where $j \in [3]$.
    \end{itemize}
    Set $F = \braces{w}$, $S = \setdef{e_i^j \in V'}{i \in [m], \, j \in [3]} \cup \setdef{v_i^j \in V'}{i \in [n], \, j \in [A]}$, and $k = n + 3m + An + 2$.
    Moreover, let $U_i = \braces{u_i^1, u_i^2, u_i^3}$ and $A_i = \setdef{v_i^j}{j \in [A]}$, for all $i \in [n]$.  
    Notice that $(G',S,F,k)$ is a path-restricted instance of \ammFVS,
    where every component of $G'[V' \setminus (S \cup F)]$ is a path on $3$ vertices.
    In \cref{fig:path_hardness} part of the construction is shown,
    assuming there exists an edge $e_i = \{ u_p, u_q \} \in E$.
    It remains to show that the two instances are equivalent in the sense that 3-colorable graphs give Yes
    instances of {\ammFVS} while finding any minimal feedback vertex set of size $k$
    in the new instance (even violating the annotations) implies that the original
    graph is 3-colorable.

    \begin{figure}[ht]
\centering
\begin{tikzpicture}[scale=0.75, transform shape]

%%%%%%%%%% vertices and text

\node[black_vertex] (vw) at (5,4) {};
\node[] () at (4.6,4) {$w$};

\node[vertex] (w1) at (2,7) {};
\node[] () at (2,7.3) {$w_1$};
\node[vertex] (w2) at (3,7) {};
\node[] () at (3,7.3) {$w_2$};
\node[vertex] (w3) at (4,7) {};
\node[] () at (4,7.3) {$w_3$};
\draw[] (vw)--(w1)--(w2)--(w3)--(vw);

\node[vertex] (w'1) at (6,7) {};
\node[] () at (6,7.3) {$w'_1$};
\node[vertex] (w'2) at (7,7) {};
\node[] () at (7,7.3) {$w'_2$};
\node[vertex] (w'3) at (8,7) {};
\node[] () at (8,7.3) {$w'_3$};
\draw[] (vw)--(w'1)--(w'2)--(w'3)--(vw);

\node[gray_vertex] (ve1) at (3,1) {};
\node[gray_vertex] (ve2) at (5,1) {};
\node[gray_vertex] (ve3) at (7,1) {};

\node[] () at (3,0.6) {$e_i^1$};
\node[] () at (5,0.6) {$e_i^2$};
\node[] () at (7,0.6) {$e_i^3$};

\node[vertex] (vu11) at (3,3) {};
\node[vertex] (vu12) at (3,4) {};
\node[vertex] (vu13) at (3,5) {};

\node[] () at (2.6,3.1) {$u_p^1$};
\node[] () at (2.6,4) {$u_p^2$};
\node[] () at (2.6,4.9) {$u_p^3$};

\node[gray_vertex] (a11) at (1,2) {};
\node[gray_vertex] (a12) at (1,6) {};

\node[] () at (0.6,2) {$v_p^1$};
\node[] () at (1,4.1) {$\vdots$};
\node[] () at (0.6,6) {$v_p^A$};

\node[vertex] (vun1) at (7,5) {};
\node[vertex] (vun2) at (7,4) {};
\node[vertex] (vun3) at (7,3) {};

\node[] () at (7.4,4.9) {$u_q^1$};
\node[] () at (7.4,4) {$u_q^2$};
\node[] () at (7.4,3.1) {$u_q^3$};

\node[gray_vertex] (a21) at (9,2) {};
\node[gray_vertex] (a22) at (9,6) {};

\node[] () at (9.4,2) {$v_q^1$};
\node[] () at (9,4.1) {$\vdots$};
\node[] () at (9.4,6) {$v_q^A$};

%%%%%%%%% edges / arcs
\draw[] (ve1)--(vw);
\draw[] (ve2)--(vw);
\draw[] (ve3)--(vw);

\draw[] (vw)--(vu11)--(vu12)--(vu13)--(vw);
\draw[] (vu11)--(ve1);
\draw[] (vu12)--(ve2);
\draw[] (vu13)--(ve3);

\draw[] (vw)--(vun1)--(vun2)--(vun3)--(vw);
\draw[] (vun1)--(ve1);
% \draw[] (vun1) edge [bend right] (vv1);
\draw[] (vun2)--(ve2);
\draw[] (vun3)--(ve3);

\draw[] (a11)--(vu11);
\draw[] (a11)--(vu12);
\draw[] (a11)--(vu13);
\draw[] (a12)--(vu11);
\draw[] (a12)--(vu12);
\draw[] (a12)--(vu13);

\draw[] (a21)--(vun1);
\draw[] (a21)--(vun2);
\draw[] (a21)--(vun3);
\draw[] (a22)--(vun1);
\draw[] (a22)--(vun2);
\draw[] (a22)--(vun3);

\end{tikzpicture}
\caption{Part of the graph $G'$ depicting the vertices associated with the edge $e_i = \braces{u_p, u_q} \in E$.
The black vertex $w$ belongs to $F$ and the gray vertices belong to $S$.}
\label{fig:path_hardness}
\end{figure}
    
    Assume that $G$ has a valid 3-coloring, say $f \colon V \to [3]$.
    Let $S' = \setdef{u_i^j \in V'}{f(u_i) = j} \cup S \cup \{w_1, w_1'\}$
    and note that the size of $S'$ is $n + 3m + An + 2 = k$,
    while $S' \supseteq S$ and $S' \cap F = \varnothing$ hold.
    $S'$ is a feedback vertex set of $G'$.
    Indeed, since $S \subseteq S'$, as well as $w_1, w_1' \in S'$,
    the only remaining cycles are due to the vertices of $U_i$ and $w$,
    for every $i \in [n]$, but $|S' \cap U_i| = 1$ holds.
    It remains to show that $S'$ is minimal.
    $S_1 = S' \setminus \{u_i^j\}$ is not a feedback vertex set, for any $u_i^j \in S'$,
    since then $w \notin S_1$ while $S \cap U_i = \varnothing$.
    Next, we argue that $S_2 = S' \setminus \{e_i^j\}$ is not a feedback vertex set, for any $e_i^j \in S'$.
    Assume that $e_i = \{u_p, u_q\}$.
    Then, since $f(u_p) \neq f(u_q)$, it holds that at least one of $u_p^j, u_q^j$ does not belong to $S'$.
    Name this vertex $x \in U_r$ where $r \in \{ p,q \}$,
    and notice that since $|S' \cap U_r| = 1$,
    there exists a path from $x$ to $w$ containing only vertices of $U_r \setminus S'$.
    In that case, since $e_i^j$ has an edge with $w$ and $x$ is a neighbor of $e_i^j$,
    it follows that $S_2$ is not a feedback vertex set.
    Moreover, $S_3 = S' \setminus \braces{v_i^j}$ is not a feedback vertex set,
    for any $v_i^j \in S'$, since then there exists a cycle composed of $v_i^j$,
    the vertices of $U_i \setminus S'$ and possibly $w$.
    Finally, $w_1$ is clearly necessary, as we have the cycle $w - w_1 - w_2 - w_3 - w$,
    and similarly for $w_1'$.

    Assume that $G'$ has a minimal feedback vertex set $S'$, with $|S'| \geq k = n+3m+An+2$.
    We first show that $w\not\in S'$, that is, $S'\cap F=\varnothing$.
    To this end, we first prove that $|S' \cap U_i| \le 1$, for all $i \in [n]$.
    If $|S' \cap U_i| \geq 2$, it follows that $S' \cap A_i = \varnothing$ since the
    vertices of $A_i$ have at most one neighbor in $V' \setminus S'$, thus they do
    not belong to $S'$.
    In that case, $S' \cap A_i = \varnothing$, implying that
    $|S'| \leq |V'| - A = 3n + 3m + An + 7 - (2n + 6) = k-1$,
    which is a contradiction.
    We therefore have that $|S' \cap U_i| \le 1$, for all $i \in [n]$.
    If $w \in S'$, then $w_1,w_2,w_3, w_1',w_2',w_3' \notin S'$,
    in which case the size of $S'$ is at most $n+3m+An+1 < k$.
    Thus it follows that $w \notin S'$, i.e., $S' \cap F = \varnothing$.
    Furthermore, $S' \cap U_i \neq \varnothing$ for all $i \in [n]$,
    since the vertices of $U_i$ and $w$ form a cycle.
    We conclude that $|S' \cap U_i| = 1$, for all $i \in [n]$.
    Now, due to the size of $S'$, as well as the fact that at most $2$ vertices
    from $\{w_1,w_2,w_3,w_1',w_2',w_3'\}$ belong to $S$, we have that $S' \supseteq S$.
    
    Now consider the coloring $f \colon V \to [3]$ where $f(u_i) = j$ if $u_i^j \in S'$.
    In that case, for $f$ to be a valid coloring, it suffices to prove
    that if $\{u_p, u_q\} \in E$, then $u_p^i, u_q^j \in S'$ for $i \neq j$.
    Assume that this is not the case, i.e.,
    there exist $u_p^i, u_q^i \in S'$ and $e_1 = \{u_p, u_q\} \in E$.
    In that case, $S' \setminus \{ e_1^i \}$ remains a feedback vertex set,
    since $e_1^i$ only has a single neighbor not belonging to $S'$, which is a contradiction.
\end{proof}

We proceed by presenting the main algorithm of this subsection, which will be
essential in proving \cref{thm:mmfvs_natural_algo}.
In the following, we say that, given an instance $\mathcal{I} = (G, S, F, k)$ of \ammFVS,
$s \in S$ is a \emph{good vertex} of $\mathcal{I}$ if $\deg_F (s) \geq 2$ and $\deg_U (s) \leq 1$.
We stress the fact that the algorithm of \cref{thm:path_restricted_algo} may return a set $S'$
that, albeit being a minimal feedback vertex set of $G$ of size at least $k$,
does not necessarily constitute a solution of the annotated instance $\mathcal{I}$,
since even though $S' \cap F = \varnothing$, it does not necessarily hold that $S' \supseteq S$.

\begin{theoremrep}\label{thm:path_restricted_algo}
    Let $\mathcal{I} = (G, S, F, k)$ be a path-restricted instance of \ammFVS,
    and let $g$ denote the number of its good vertices.
    There is an algorithm running in time $3^{k-g} n^{\bO(1)}$ which either returns
    a minimal feedback vertex set $S'$ of $G$ of size $|S'| \geq k$ such that $S' \cap F = \varnothing$,
    or concludes that $\mathcal{I}$ is a No instance of \ammFVS.
\end{theoremrep}

\begin{proof}
    Let, for an instance $\mathcal{I} = (G, S, F, k)$ of \ammFVS,
    $Z(\mathcal{I}) \subseteq S$ denote the subset of \emph{marked} vertices of $S$,
    which is comprised of the vertices of $S$ that have two neighbors in $F$ belonging to the same
    connected component of $G[F \cup U]$.
    The main idea behind the algorithm lies in the fact that we can efficiently handle instances where
    $|Z(\mathcal{I})| \geq k$ or $Z(\mathcal{I}) = S$.
    Towards this, we will employ a branching strategy that, as long as this is not the case,
    produces new instances $\mathcal{I}'$ such that $Z(\mathcal{I}') \supset Z(\mathcal{I})$.
    Prior to performing branching, we first observe that we can efficiently deal with the good vertices.
    Afterwards, by employing said branching strategy, in every step we decide which additional vertex will be marked,
    thereby increasing the number of marked vertices on each iteration.
    If at some point $|Z(\mathcal{I})| \geq k$ or $Z(\mathcal{I}) = S$,
    it remains to decide whether this comprises a viable solution $S'$.
    We start with \cref{lem:connectivity_in_solution}.

    \begin{lemma}\label{lem:connectivity_in_solution}
        Let $\mathcal{I} = (G, S, F, k)$ be a path-restricted instance of {\ammFVS} and $S^*$
        a minimal feedback vertex set of $G$, where $S^* \supseteq S$ and $S^* \cap F = \varnothing$,
        and $F^* = V(G) \setminus S^*$ denotes the corresponding forest.
        \begin{romanenumerate}
            \item From every path of $G[U]$, at most one vertex belongs to $S^*$.
            \item Let $u,v \in F^*$. Then, $u$ and $v$ are in the same connected component of $G[F \cup U]$
            if and only if they are in the same connected component of $G[F^*]$.
        \end{romanenumerate}
    \end{lemma}
    
    \begin{nestedproof}
        For the first statement, suppose there exist distinct $u_1,u_2 \in S^* \cap U$ belonging to the same
        path of $G[U]$, where $P \subseteq U$ denotes the set of vertices of said path.
        In that case, $G[F^* \cup \braces{u_1}]$ must contain a cycle involving $u_1$.
        Since $\mathcal{I}$ is a path-restricted instance, it holds that $\forall v \in P$,
        $\deg_{F \cup U} (v) = 2$,
        and since $F^* \cup \braces{u_1} \subseteq F \cup U$,
        it follows that $\deg_{F^* \cup \braces{u_1}} (v) \leq 2$.
        Therefore, for $G[F^* \cup \braces{u_1}]$ to contain such a cycle it holds that
        $F^* \supseteq P \setminus \braces{u_1}$,
        which is a contradiction.
    
        For the second statement,
        first consider the case when $u, v \in F$, both belonging to the same connected component of $G[F \cup U]$.
        Let $P$ be a path of $G[F \cup U]$ connecting $u$ and $v$,
        where $f_1, \ldots, f_j$ are the vertices of $P$ belonging to $F$ in the order that they appear in $P$,
        i.e., $f_1 = u$ and $f_j = v$.
        We claim that any two consecutive vertices $f_i, f_{i+1}$ belong to the same connected component of $G[F^*]$.
        Fix some $i \in [j-1]$ and notice that if $f_i, f_{i+1}$ are adjacent,
        both vertices belong to the same connected component of $G[F]$,
        and since $F^* \supseteq F$, that is also the case in graph $G[F^*]$.
        Assume otherwise, i.e., that $f_i, f_{i+1}$ are not adjacent.
        Since no vertex of $F$ appears between $f_i $ and $f_{i+1}$ in $P$,
        it follows that there exists a path of $U$ whose endpoints are adjacent to $f_i$ and $f_{i+1}$ respectively. 
        Now, for this path, either all of its vertices belong to $F^*$, or one of its vertices, say $w$, is in $S^*$.
        In the first case, $f_i$ and $f_{i+1}$ are in the same connected component of $G[F^*]$ due to said path.
        In the latter case, the private cycle of $w$ in $G[F^* \cup \braces{w}]$ contains both $f_i$ and $f_{i+1}$,
        thus they are in the same connected component of $G[F^*]$.
        Consequently, the claim holds, and due to the transitivity of connectivity,
        it follows that $u$ and $v$ belong to the same connected component of $G[F^*]$.
    
        In case at least one of $u, v$ belongs to $U$, let $F' = F \cup \braces{u,v}$
        and consider the instance $\mathcal{I}' = (G, S, F', k)$.
        Obviously, $u,v$ are in the same connected component of $G[F \cup U]$ if and only if they are in the same connected component of $G[F' \cup U']$,
        where $U' = U \setminus \braces{u,v}$.
        Moreover, any $S^*$ that does not contain $u,v$ is a solution of instance $\mathcal{I}$ if and only if it is a solution of $\mathcal{I}'$.
        Thus, the statement follows.

        For the converse direction, it suffices to notice that $F^* \subseteq F \cup U$.
    \end{nestedproof}

    Due to \cref{lem:connectivity_in_solution} we can thus infer the connected components of any
    forest $F^* = V \setminus S^*$,
    where $S^*$ is a minimal feedback vertex set of $G$ such that $S^* \supseteq S$ and $S^* \cap F = \varnothing$.
    Moreover, we can apply the following reduction rule,
    whose correctness follows from \cref{lem:connectivity_in_solution}.
    
    \subparagraph*{Rule $(\diamond)$.}
    Let $\mathcal{I} = (G, S, F, k)$ be a path-restricted instance of \ammFVS,
    and $u \in U$ such that the connected components of $G[(F \cup U) \setminus \braces{u}]$
    are more than the connected components of $G[F \cup U]$.
    Then, replace $\mathcal{I}$ with $\mathcal{I}' = (G, S, F \cup \braces{u}, k)$.
    
    We next handle the cases where for instance $\mathcal{I}$ it holds that either $|Z(\mathcal{I})| \geq k$ or $Z(\mathcal{I}) = S$.
    We consider the two cases separately;
    notice that when $|Z(\mathcal{I})| \geq k$, the returned set $S'$ does not necessarily respect the annotation,
    since although $|S'| \geq k$ and $S' \cap F = \varnothing$, it does not necessarily hold that $S' \supseteq S$.

    \begin{lemma}\label{lem:path_restricted_base_case1}
        Let $\mathcal{I} = (G, S, F, k)$ be a path-restricted instance of \ammFVS,
        where $|Z(\mathcal{I})| \geq k$.
        There exists an algorithm that in polynomial time returns a minimal feedback vertex set $S'$ of $G$
        such that $|S'| \geq k$ and $S' \cap F = \varnothing$.
    \end{lemma}

    \begin{nestedproof}
        Initially set $S' = S$.
        Greedily add any vertex $u \in U$ to the forest, as long as no cycles are formed;
        alternatively, add $u$ to $S'$, in which case notice that the vertices of $F$ adjacent to the path
        $u$ belongs to are in the same connected component.
        Lastly, greedily remove vertices of $S'$ so that it becomes minimal.
        In the end, for $F' = V(G) \setminus S'$ it holds that $F' \supseteq F$,
        which implies that $S' \cap F = \varnothing$,
        while any two vertices $f_1, f_2 \in F$ belonging to the same connected component of $G[F \cup U]$
        belong to the same connected component of $G[F']$,
        thus every vertex of $Z(\mathcal{I})$ has a private cycle,
        i.e., $S' \supseteq Z(\mathcal{I})$,
        implying that $|S'| \geq k$.
    \end{nestedproof}

    \begin{lemma}\label{lem:path_restricted_base_case2}
        Let $\mathcal{I} = (G, S, F, k)$ be a path-restricted instance of \ammFVS,
        where $Z(\mathcal{I}) = S$.
        There exists an algorithm that decides $\mathcal{I}$ in polynomial time,
        and in case it is a Yes instance it returns a minimal feedback vertex set $S'$ of $G$
        of size $|S'| \geq k$, where $S' \supseteq S$ and $S' \cap F = \varnothing$.
    \end{lemma}

    \begin{nestedproof}
        Recall that due to \cref{lem:connectivity_in_solution},
        it holds that for any minimal feedback vertex set $S^*$ where $S^* \cap F = \varnothing$ and $S^* \supseteq S$,
        if $u,v \in F^*$, where $F^* = V(G) \setminus S^*$,
        then $u$ and $v$ are in the same connected component of $G[F^*]$ if and 
        only if that is the case in $G[F \cup U]$.
        We will say that a path of $U$ belongs to $F^*$ when all of its vertices belong to $F^*$.
        
        Notice that the vertices of $F$ can be partitioned into equivalence classes, depending on their connectivity 
        in $G[F \cup U]$.
        For $u, v \in F$, let them belong to the same equivalence class $C_i$, for $i \in [p]$,
        if they are in the same connected component of $G[F \cup U]$,
        where $p \leq |F|$ denotes the number of equivalence classes.
        Now, for each $C_i$, let $c_i = \cc(G[C_i])$ be equal to the number of connected components of $G[C_i]$.
        All components of $G[C_i]$ must be connected in $G[F^*]$,
        thus the number of paths required is $c_i - 1$ per equivalence class $C_i$,
        since each path of $U$ either reduces the number of connected components of $G[C_i]$ by exactly $1$,
        or induces a cycle.
        Therefore, it suffices to greedily add each path to the final forest $F'$, as long as no cycle is formed.
        In case a cycle is indeed formed,
        then due to \cref{lem:connectivity_in_solution} it suffices to add one of its vertices to $S'$,
        since it has two edges towards the same connected component of $F'$.
        In the end, two vertices are connected in $G[F']$ if that is the case in $G[F \cup U]$ as well,
        thus every vertex of $Z(\mathcal{I}) = S$ has a private cycle since $F' \supseteq F$,
        while $S'$ is a minimal feedback vertex set such that $S' \supseteq S$ and $S' \cap F = \varnothing$.
        To see that $S'$ is an optimal solution of the annotated instance, notice that each path of $\mathcal{I}$
        either reduces by $1$ the connected components of $G[F]$, or increases by $1$ the cardinality of $S'$.
        In that case, we can determine whether $\mathcal{I}$ is a Yes or No instance,
        depending on whether $|S'| \geq k$ holds.
    \end{nestedproof}

    Armed with \cref{lem:path_restricted_base_case1,lem:path_restricted_base_case2},
    we are now ready to describe our algorithm.
    Let $\mathcal{I} = (G, S, F, k)$ be a path-restricted instance of \ammFVS.
    Notice that if at any point of execution of our algorithm
    there exists some vertex $s \in S$ which does not have two edges towards the same connected component of $G[F \cup U]$,
    then this is a No instance of {\ammFVS} and we discard it.
    Moreover, we exhaustively apply Rule $(\diamond)$ in every produced instance,
    thus inducing a polynomial-time overhead.
    
    Regarding our branching strategy, we consider the different cases for vertices of $U$.
    Notice that when a vertex $u \in U$ is moved from $U$ to $S$,
    due to \cref{lem:connectivity_in_solution} it is imperative that the connectivity of the vertices belonging to $F$
    with respect to $G[F \cup U]$ remains the same with respect to $G[(F \cup U) \setminus \braces{u}]$.
    Since we have assumed that Rule $(\diamond)$ has been exhaustively applied, that is indeed the case.
    We will first do some preprocessing and afterwards describe a branching strategy which,
    as long as $Z(\mathcal{I}) \subset S$ and $|Z(\mathcal{I})|< k$,
    marks at least one extra vertex per step.
    Since the connectivity of the vertices of $F$ in the new instances remains unchanged,
    the already marked vertices remain so.
    
    \subparagraph*{Preprocessing.}
    Suppose there exists a good vertex $h \in S$ such that $h \notin Z(\mathcal{I})$,
    i.e., all of its neighbors in $F$ belong to distinct connected components of $G[F \cup U]$.
    Recall that $h$ has at most one neighbor in $U$.
    In that case, for $h$ to have a private cycle,
    it is necessary that its neighbor $u \in N(h) \cap U$ belongs to the forest,
    as well as that it is in the same connected component of $G[F \cup U]$ as one of the other neighbors of $h$ in $F$
    (if there is no such neighbor of $h$ in $F$ we can discard the instance).
    Therefore, we replace $\mathcal{I}$ with $\mathcal{I}' = (G, S, F \cup \braces{u}, k)$,
    where $Z(\mathcal{I}') \supseteq Z(\mathcal{I}) \cup \braces{h}$.
    Note that the preprocessing can be done in polynomial time
    while for the resulting instance $\mathcal{I}^* = (G^*, S^*, F^*, k^*)$ it holds that $|Z(\mathcal{I}^*)| \geq g$,
    where $g$ denotes the number of good vertices of $\mathcal{I}$.

    \subparagraph*{Branching.}
    Let $\mathcal{I} = (G, S, F, k)$ be the instance after the preprocessing.
    Consider a vertex $s \in S \setminus Z(\mathcal{I})$.
    For $u \in U$, let $P_u \subseteq U \setminus \braces{u}$ denote the vertices in the same path as $u$ in $G[U]$.
    Consider the following cases: either there exists $u \in N(s) \cap U$ such that $u$ is in the same connected component of $G[F \cup U]$ as some $f \in N(s) \cap F$ or not.
    \begin{itemize}
        \item In the first case, we branch depending on whether $u$ is in the feedback vertex set or not. 
        Notice that if $u$ is in the feedback vertex set, then all vertices of $P_u$ must be in the forest due to \cref{lem:connectivity_in_solution}.
        Therefore, we replace our current instance with the following two:
        \begin{itemize}
            \item $\mathcal{I}_1 = (G, S\cup \braces{u}, F \cup P_u, k)$, and 
            \item $\mathcal{I}_2 = (G, S , F \cup \braces{u}, k)$.
        \end{itemize}
        In both instances we mark at least one extra vertex:
        $u$ in $\mathcal{I}_1$ and $s$ in $\mathcal{I}_2$.
    
        \item In the latter case, two vertices $a,b \in N(s) \cap U$ that belong to the same connected component of $G[F \cup U]$ must exist.
        For these vertices we branch on the following 3 cases:
        $a,b \in F$, or $a \in S$, or $b \in S$.
        Therefore, we replace the current instance with the following three:     
        \begin{itemize}
            \item $\mathcal{I}_1 = (G, S, F \cup \braces{a,b}, k)$,
            \item $\mathcal{I}_2 = (G, S \cup \braces{a}, F \cup P_a, k)$,
            \item $\mathcal{I}_3 = (G, S \cup \braces{b}, F \cup P_b, k)$.
        \end{itemize}
        In each of these instances we mark at least one extra vertex:
        $s$ in $\mathcal{I}_1$, $a$ in $\mathcal{I}_2$, and $b$ in $\mathcal{I}_3$.
    \end{itemize}
    
    \subparagraph*{Complexity.}
    The preprocessing part of the algorithm, as well as the application of the rules, require polynomial time.
    The branching strategy previously described results in at most $3^{k-g}$ instances,
    since on every step at most $3$ instances may be produced, while the branching ceases when $|Z(\mathcal{I})| \geq k$.
    Lastly, due to \cref{lem:path_restricted_base_case1,lem:path_restricted_base_case2},
    the cases when $|Z(\mathcal{I})| \geq k$ or $Z(\mathcal{I}) = S$
    are solvable in polynomial time.
    Therefore, the final running time is $3^{k-g} n^{\bO(1)}$.
    \end{proof}

% \begin{proofsketch}
%     The main idea of the algorithm lies on the fact that we can efficiently handle instances where either $k=0$ or $S=\varnothing$.
%     Towards this, we will employ a branching strategy that, as long as $S$ remains non empty, new instances with reduced $k$ are produced.
%     Prior to performing branching, we first observe that we can efficiently deal with the good vertices.
%     Afterwards, by employing said branching strategy, in every step we decide which vertex will be counted
%     towards the $k$ required, thereby reducing parameter $k$ on each iteration.
%     If at some point $k=0$ or $S=\varnothing$, it remains to decide whether this comprises a viable solution $S'$.
%     Notice that $S'$ may not be a solution for the annotated instance,
%     since even if $|S'| \geq k$, it does not necessarily hold that $S' \supseteq S$.
% \end{proofsketch}

\subsection{Algorithm for Max Min FVS}\label{subsec:general_algorithm}
We start by presenting a high level sketch of the algorithm for \mmFVS.
The starting point is a minimal feedback vertex set $S_0$ of $G$.
Note that such a set can be obtained in polynomial time,
while if it is of size at least $k$ we are done.
Therefore, assume that $|S_0| < k$.
Then, assuming there exists a minimal feedback vertex set $S^*$, where $|S^*| \geq k$ and $F^* = V(G) \setminus S^*$,
we will guess $S_0 \cap S^*$, thereby producing instances $\mathcal{I}_0 = (G, S_0 \cap S^*, S_0 \cap F^*, k)$ of \ammFVS.
Subsequently, we will establish a number of \emph{safe} reduction rules, which do not affect the answer of the instances.
We will present a measure of progress $\mu$, which guarantees that if an instance $\mathcal{I} = (G, S, F, k)$ of {\ammFVS} has $\mu(\mathcal{I}) \leq 1$,
then $G$ has a minimal feedback vertex set $S'$ of size at least $k$, where $S' \cap F = \varnothing$.
Then, we will employ a branching strategy which, given $\mathcal{I}_i$,
will produce instances $\mathcal{I}^1_{i+1}, \mathcal{I}^2_{i+1}$ of lesser measure of progress,
such that $\mathcal{I}_i$ is a Yes instance if and only if at least one of $\mathcal{I}^1_{i+1}, \mathcal{I}^2_{i+1}$
is also a Yes instance.
If we can no further apply our branching strategy, and the measure of progress remains greater than $1$,
then it holds that $\mathcal{I}$ is a path-restricted instance and \cref{thm:path_restricted_algo} applies.

\subparagraph*{Measure of progress.}
Let $\mathcal{I} = (G, S, F, k)$ be an instance of \ammFVS.
We define as $\mu(\mathcal{I}) = k + \cc(G[F]) - g - p$ its \emph{measure of progress}, where
\begin{itemize}
    \item $\cc(G[F])$ denotes the number of connected components of $G[F]$,
    
    \item $g$ denotes the number of good vertices of $\mathcal{I}$, i.e., vertices $s \in S$ such that
    $\deg_F (s) \geq 2$ and $\deg_U (s) \leq 1$,
    
    \item $p$ denotes the number of \emph{interesting paths} of $G[U]$,
    where a connected component of $G[U]$ is an interesting path if for
    every vertex $u$ belonging to said component, it holds that $\deg_{F \cup U} (u) = 2$.
    Notice that if every connected component of $G[U]$ is an interesting path,
    then $\mathcal{I}$ is a path-restricted instance.
\end{itemize}
It holds that if $\mu(\mathcal{I}) \leq 1$, then the underlying {\mmFVS} instance has a positive answer,
which does not necessarily respect the constraints dictated by the annotated version.

\begin{lemmarep} \label{lem:end_of_progress}
    Let $\mathcal{I} = (G, S, F, k)$ be an instance of \ammFVS,
    where $\mu(\mathcal{I}) \leq 1$.
    Then, $G$ has a minimal feedback vertex set $S'$ of size at least $k$, where $S' \cap F = \varnothing$.
\end{lemmarep}

\begin{proof}
    Since $F$ is a forest, $S \cup U$ comprises a valid feedback vertex set of $G$.
    Let $S'$ be a minimal feedback vertex set obtained in polynomial time from $S \cup U$,
    while $F' = V \setminus S'$ denotes the forest resulting from the vertices belonging to $F$ plus
    the vertices of $(S \cup U) \setminus S'$.

    Let a \textit{loss} be when either a good vertex of $S$,
    or the entirety of an interesting path belongs to $F' = V \setminus S'$.
    Notice that both good vertices and interesting paths have at least $2$ edges to some vertices of $F$.
    Consequently, for every loss, the connected components of $F$ reduce by at least $1$:
    in order to move a good vertex or an interesting path to the forest, no cycles should be formed,
    i.e., all of their neighbors are in distinct connected components of $F$,
    thus the connected components of the forest will be reduced.
    Therefore, it follows that at most $\cc(G[F]) - 1$ losses may happen,
    which means that $S'$ contains at least $g + p - (\cc(G[F]) - 1)$ vertices;
    each of those corresponds to either a good vertex or belongs to an interesting path which does not belong entirely to $F'$.
    In that case however, $|S'| \geq g + p - (\cc(G[F]) - 1) \geq k$, since $\mu(\mathcal{I}) \leq 1$.
\end{proof}

Next, we describe some reduction rules which neither affect the answer of an instance of \ammFVS,
nor increase its measure of progress.
For their correctness, we repeatedly use \cref{lem:contraction} which we first present.

\begin{lemmarep} \label{lem:contraction}
    Let $G = (V,E)$ be a (multi)graph and $uv \in E(G)$.
    Then, $G$ is acyclic if and only if $G / uv$ is acyclic. 
\end{lemmarep}

\begin{proof}
    First consider the case where there are multiple edges between $u$ and $v$. 
    Then, $G$ has a cycle that uses these edges while $G / uv$ has a self loop and the statement holds.
    Thus, it suffices to consider the case where there is only a single edge between $u$ and $v$,
    in which case the vertex $w$ that has replaced $u$ and $v$ in $G/ uv$ does not have a self loop in $G/ uv$.
    
    For the forward direction, assume that $G$ has a cycle.
    Notice that any cycles not including the edge $uv$ are not affected by its contraction,
    thus assume that $uv$ is part of a cycle in $G$.
    Since $G$ does not include any edges parallel to $uv$, this cycle has at least three vertices.
    This means that there exists a path from $u$ to $v$ which does not include the edge $uv$.
    Then, in $G / uv$, this path is a cycle as we have replaced $u$ and $v$ with a single vertex.

    For the converse direction, assume that $G/ uv$ has a cycle $C$
    and let $w$ be the vertex that has replaced $u$ and $v$ in $G/ uv$.
    There are two cases, either $w \notin V(C)$ or $w \in V(C)$.
    In the first case notice that $C$ is also a cycle in $G$ therefore the statement holds.
    In the latter, since $w$ does not have a self loop, 
    there is a path $P$ consisting of at least $1$ vertex such that,
    the first and the last vertex of said path are both adjacent to $w$. 
    Let $v_s$ and $v_t$ denote those (not necessarily distinct) vertices. 
    If there is $v' \in \{u,v\}$ such that $v' \in N(v_s) \cap N(v_t)$ then 
    the path $P$ together with $v'$ comprises a cycle in $G$. 
    Otherwise, one among $v_s, v_t$ is adjacent to $u$ and the other to $v$. 
    Without loss of generality let $v_su,v_tv \in E(G)$. Notice that there is a path in $G$ that starts with $u$, 
    ends with $v$, and uses the vertices in $P$.
    Consequently, this path does not include the edge $uv$.
    Adding the edge $uv$ to this path results in a cycle in $G$.
\end{proof}

\subparagraph*{Rule 1.}
Let $\mathcal{I} = (G,S,F,k)$ be an instance of \ammFVS, $u, v \in F$, and $uv \in E(G)$.
Then, replace $\mathcal{I}$ with $\mathcal{I}' = (G',S,F',k)$,
where $G' = G / uv$ occurs from the contraction of $u$ and $v$ into $w$,
while $F' = (F \cup \braces{w}) \setminus \braces{u,v}$.

\subparagraph*{Rule 2.}
Let $\mathcal{I} = (G,S,F,k)$ be an instance of \ammFVS, $u \in U$, and $\deg_{F \cup U} (u) = 0$.
Then, replace $\mathcal{I}$ with $\mathcal{I}' = (G-u,S,F,k)$.

\subparagraph*{Rule 3.}
Let $\mathcal{I} = (G,S,F,k)$ be an instance of \ammFVS, $u \in U$, and $\deg_{F \cup U} (u) = 1$,
where $N(u) \cap (F \cup U) = \{ v \}$.
Then, replace $\mathcal{I}$ with $\mathcal{I}' = (G',S,F',k)$,
where $G' = G / uv$ occurs from the contraction of $u$ and $v$ into $w$,
while $F' = (F \cup \braces{w}) \setminus \braces{v}$ if $v \in F$,
and $F' = F$ otherwise.

\begin{lemmarep}
    Applying rules 1, 2, and 3 is safe and
    does not increase the measure of progress.
\end{lemmarep}

\begin{proof}
    We will provide a proof for each rule in a distinct paragraph.
    
    \proofsubparagraph*{Rule 1.}
    Let $\mathcal{I} = (G,S,F,k)$ be an instance of {\ammFVS} and $\mathcal{I}' = (G',S,F',k)$
    the instance of {\ammFVS}
    resulting from applying Rule 1 to $\mathcal{I}$,
    where $G' = (V', E')$ occurs from the contraction of $u$ and $v$ into $w$ (i.e., $G' = G / uv$),
    while $F' = (F \cup \braces{w}) \setminus \braces{u,v}$.
    We will show that $\ammfvs(\mathcal{I}') = \ammfvs(\mathcal{I})$ and $\mu(\mathcal{I}') \leq \mu(\mathcal{I})$.

    Let $S_1 \supseteq S$ be a minimal feedback vertex set of $G$ such that $S_1 \cap F = \varnothing$,
    which implies that $u,v \notin S_1$.
    We claim that $S_1$ is a minimal feedback vertex set of $G'$.
    Let $F_1 = V \setminus S_1$ and $F_2 = V' \setminus S_1$.
    By deﬁnition, $G[F_1]$ is acyclic.
    Due to \cref{lem:contraction} and the fact that $G'[F_2]$ is obtained from $G[F_1]$ by contracting $uv$,
    it follows that $G'[F_2]$ is also a forest.
    To see that $S_1$ is a \emph{minimal} feedback vertex set of $G'$,
    let $z \in S_1$ and consider the graphs $G_1 = G [F_1 \cup \braces{z}]$
    and $G_2 = G'[F_2 \cup \braces{z}]$.
    Notice that $G_2$ can be obtained from $G_1$ by contracting $uv$.
    Moreover, due to the minimality of $S_1$, $G_1$ contains a cycle,
    and due to \cref{lem:contraction} it follows that $G_2$ also does.
    Thus, $S_1$ is a minimal feedback vertex set of $G'$,
    and $\ammfvs(\mathcal{I}) \leq \ammfvs(\mathcal{I}')$ follows.

    For the converse direction, let $S_2 \supseteq S$ be a minimal feedback vertex set of $G'$
    such that $S_2 \cap F' = \varnothing$, which implies that $w \notin S_2$.
    We claim that $S_2$ is a minimal feedback vertex set of $G$.
    Let $F_1 = V \setminus S_2$ and $F_2 = V' \setminus S_2$.
    By deﬁnition, $G'[F_2]$ is acyclic.
    Due to \cref{lem:contraction} and the fact that $G'[F_2]$ is obtained from $G[F_1]$ by contracting $uv$,
    it follows that $G[F_1]$ is also a forest.
    To see that $S_2$ is a \emph{minimal} feedback vertex set of $G$,
    let $z \in S_2$ and consider the graphs $G_1 = G [F_1 \cup \braces{z}]$
    and $G_2 = G'[F_2 \cup \braces{z}]$.
    Notice that $G_2$ can be obtained from $G_1$ by contracting $uv$.
    Moreover, due to the minimality of $S_2$, $G_2$ contains a cycle,
    and due to \cref{lem:contraction} it follows that $G_1$ also does.
    Thus, $S_2$ is a minimal feedback vertex set of $G$,
    and $\ammfvs(\mathcal{I}) \geq \ammfvs(\mathcal{I}')$ follows.
    
    Moreover, it holds that $\mu(\mathcal{I}') = \mu(\mathcal{I})$,
    since $\cc(G[F]) = \cc(G'[F'])$, while the number of interesting paths
    and good vertices of both instances is the same.

    \proofsubparagraph*{Rule 2.}
    Let $\mathcal{I} = (G,S,F,k)$ be an instance of {\ammFVS} and
    $\mathcal{I}' = (G',S,F,k)$ the instance of {\ammFVS} resulting from applying Rule 2 to $\mathcal{I}$,
    where $G' = (V', E')$ occurs from the deletion of some $u \in U$
    with $\deg_{F \cup U} (u) = 0$ (i.e., $G' = G - u$).
    We will show that $\ammfvs(\mathcal{I}') = \ammfvs(\mathcal{I})$ and 
    $\mu(\mathcal{I}') \leq \mu(\mathcal{I})$.
    
    Let $S_1 \supseteq S$ be a minimal feedback vertex set of $G$ such that $S_1 \cap F = \varnothing$.
    Since $N_G(u) \subseteq S \subseteq S_1$,
    it follows that $u \notin S_1$ as otherwise
    $S_1 \setminus \braces{u}$ remains a feedback vertex set and $S_1$ is not minimal.
    In that case, $S_1$ is a feedback vertex set of $G - u$.
    To see that $S_1$ is a \emph{minimal} feedback vertex set of $G - u$,
    notice that for any $z \in S_1$, $G[F \cup \{z\}]$ contains a cycle which does not include $u$,
    therefore this cycle is also present in $G'[(F \cup \{z\}) \setminus \{u\}]$,
    i.e., $z$ has a private cycle in $G - u$.
    Thus, $S_1$ is a minimal feedback vertex set of $G - u$,
    and $\ammfvs(\mathcal{I}) \leq \ammfvs(\mathcal{I}')$ follows.

    For the converse direction, let $S_2 \supseteq S$ be a minimal feedback vertex set of $G - u$
    such that $S_2 \cap F = \varnothing$.
    We claim that $S_2$ is a minimal feedback vertex set of $G$.
    We first show that $G[V \setminus S_2]$ is a forest.
    Notice that $V \setminus S_2 \subseteq F \cup U$,
    thus $N_G(u) \cap (V \setminus S_2) = \varnothing$,
    which implies that no cycle present in $G[V \setminus S_2]$ contains $u$.
    Now, assume that $G[V \setminus S_2]$ contains a cycle.
    Since $u$ does not belong to any such cycle,
    it follows that $G[V \setminus (S_2 \cup \{u\})]$ contains a cycle as well,
    which is a contradiction since $S_2$ is a feedback vertex set of $G-u$.
    For the minimality, notice that since all vertices of $S_2$ have a private cycle in $G - u$,
    that is the case for graph $G$ as well.
    Therefore, $\ammfvs(\mathcal{I}) \geq \ammfvs(\mathcal{I}')$.

    Lastly, $\mu(\mathcal{I}') \leq \mu(\mathcal{I})$,
    since the deletion of $u$ affects neither $\cc(G[F])$ nor the number of interesting paths,
    while the number of good vertices may increase.

    \proofsubparagraph*{Rule 3.}
    Let $\mathcal{I} = (G,S,F,k)$ be an instance of {\ammFVS} and
    $\mathcal{I}' = (G',S,F',k)$ the instance of {\ammFVS} resulting from applying Rule 3 to $\mathcal{I}$,
    where $G' = (V', E')$ occurs from the contraction of $u$ and $v$ into $w$ (i.e., $G' = G / uv$),
    for some $u \in U$ such that $\deg_{F \cup U} (u) = 1$, where $N(u) \cap (F \cup U) = \{ v \}$.
    Moreover, it holds that $F' = (F \cup \braces{w}) \setminus \braces{v}$ if $v \in F$,
    and $F' = F$ otherwise.
    We will show that $\ammfvs(\mathcal{I}') = \ammfvs(\mathcal{I})$ and
    $\mu(\mathcal{I}') \leq \mu(\mathcal{I})$.

    Let $S_1 \supseteq S$ be a minimal feedback vertex set of $G$
    such that $S_1 \cap F = \varnothing$.
    Notice that since $\deg_{F \cup U} (u) = 1$, it holds that $u \notin S_1$ as otherwise $S_1$ is not minimal.

    We first argue that the statement is true when $v \in F$.
    Indeed, since $u \notin S_1$, it holds that $\mathcal{I}$ is a Yes instance of {\ammFVS} if and only if
    $\mathcal{J} = (G,S,F \cup \{u\},k)$ is a Yes instance of \ammFVS.
    We remark that $\mu(\mathcal{J}) \le \mu(\mathcal{I})$,
    since the two instances have the same number of interesting paths,
    $\cc(G[F]) = \cc(G[F \cup \{ u \}])$,
    and $\mathcal{J}$ has at least as many good vertices as $\mathcal{I}$. 
    Lastly, notice that applying Rule 1 on $\mathcal{J}$ results in instance $\mathcal{I}'$.
    
    It remains to prove the statement when $v \in U$.
    For the forward direction, let $S_1$ as before,
    and consider the cases $v \notin S_1$ and $v \in S_1$.

    If $v \notin S_1$, then we claim that $S_1$ is a minimal feedback vertex set of $G'$.
    Notice that contracting $uv$ in $G[V \setminus S_1]$ results in $G'[V' \setminus S_1]$,
    thus due to \cref{lem:contraction} the latter is a forest.
    Furthermore, for all $z \in S_1$,
    contracting $uv$ in $G[(V \setminus S_1) \cup \braces{z}]$ results in $G'[(V' \setminus S_1) \cup \braces{z}]$,
    thus due to \cref{lem:contraction} the latter has a cycle.
    Consequently, $S_1$ is a minimal feedback vertex set of $G'$.

    If $v \in S_1$, then we claim that $S^*_1 = (S_1 \setminus \braces{v}) \cup \braces{w}$ is a minimal feedback vertex set of $G'$.
    Notice that $S^*_1$ is a feedback vertex set of $G'$, since $S_1 \cup \braces{u}$ is a feedback vertex set of $G$.
    To see that it is minimal, for all $z \in S^*_1 \setminus \braces{w}$ we observe that $G'[(V' \setminus S^*_1) \cup \braces{z}]$ is
    obtained from $G[(V \setminus S_1) \cup \braces{z}]$ by deleting $u$,
    which has degree at most $1$ due to $z$.
    Therefore, this deletion does not destroy any cycles, while $G[(V \setminus S_1) \cup \braces{z}]$ contains
    some cycle due to the minimality of $S_1$.
    Finally, for $w \in S^*_1$,
    observe that contracting $uv$ in $G[(V \setminus S_1) \cup \braces{v}]$, which contains a cycle,
    results in $G'[(V' \setminus S^*_1) \cup \braces{w}]$,
    thus due to \cref{lem:contraction} the latter contains a cycle as well.

    Consequently, $\ammfvs(\mathcal{I}) \leq \ammfvs(\mathcal{I}')$ follows.
    For the converse direction,
    let $S_2 \supseteq S$ be a minimal feedback vertex set of $G'$ such that $S_2 \cap F' = \varnothing$.
    Recall that we consider the case where $u,v \in U$, so $F'=F$.  
    We consider two cases, either $w \notin S_2$ or $w \in S_2$.

    If $w \notin S_2$, we claim that $S_2$ is a minimal feedback vertex set of $G$.
    Notice that contracting $uv$ in $G[V \setminus S_2]$ results in $G'[V'\setminus S_2]$,
    thus due to \cref{lem:contraction} the latter is acyclic and $S_2$ is a feedback vertex set of $G$.
    Regarding minimality, let $z \in S_2$.
    Since contracting $uv$ in $G[(V \setminus S_2) \cup \{z\}]$
    results in $G'[(V' \setminus S_2) \cup \{z\}]$,
    due to \cref{lem:contraction} it follows that the latter has a cycle.
    Therefore, $S_2$ is a minimal feedback vertex set of $G$.

    If $w \in S_2$, we claim that $S^*_2 = (S_2 \cup \{v\}) \setminus \{w\}$ is a minimal feedback vertex set of $G$.
    Let $F_2 = V' \setminus S_2$ and $F^*_2 = V \setminus S^*_2$.  
    Notice that in $G[F^*_2]$, $u$ is an isolated vertex since $N_G (u) \cap (F \cup U) = \{ v \} \subseteq S^*_2$. 
    Moreover, $G[F^*_2 \setminus \{u\}]$ is acyclic since it is the same as $G'[F_2]$.
    Therefore, $S^*_2$ is a feedback vertex set of $G$.
    It remains to show that $S^*_2$ is minimal.
    Let $z \in S_2 \setminus \braces{w}$. 
    Notice that in $G[F^*_2 \cup \{z\}]$, $u$ has degree at most $1$ due to $z$, 
    therefore it is not contained in any cycle of $G[F^*_2 \cup \{z\}]$. 
    Consequently, $G[F^*_2 \cup \{z\}]$ contains a cycle if and only if
    $G[(F^*_2 \setminus \{u\}) \cup \{z\}]$ contains a cycle as well. 
    However, $G[(F^*_2 \setminus \{u\}) \cup \{z\}]$ has a cycle as it is the same as
    $G'[F_2 \cup \{z\}]$. 
    As for $v$, notice that contracting $uv$ in $G[F^*_2 \cup \{v\}]$ results in $G'[F_2 \cup \{w\}]$.
    Therefore, by \cref{lem:contraction} and the minimality of $S_2$,
    it follows that $G[F^*_2 \cup \{v\}]$ has a cycle and
    thus $S^*_2$ is a minimal feedback vertex set of $G$.
    
    Consequently, $\ammfvs(\mathcal{I}') = \ammfvs(\mathcal{I})$ follows.
    Lastly, we need to show that $\mu(\mathcal{I}') \leq \mu(\mathcal{I})$ in the case where $u,v\in U$.
    Indeed, if $u,v \in U$, then the number of components in $F$ remains unchanged,
    while the number of interesting paths and of good vertices in $S$ does not decrease
    by the contraction of $uv$.
\end{proof}

After exhaustively applying the aforementioned rules,
it holds that $\forall u \in U$, $\deg_{F \cup U} (u) \geq 2$,
thus $G[U]$ is a forest containing trees, 
all the leaves of which have at least one edge to $F$.
Moreover, $F$ is an independent set.
We proceed with a branching strategy that produces instances of {\ammFVS} of reduced measure of progress.
If at some point $\mu(\mathcal{I}) \leq 1$ for some produced instance $\mathcal{I}$,
then \cref{lem:end_of_progress} can be applied.

\subparagraph*{Branching strategy.}
Let $\mathcal{I} = (G, S, F, k)$ be an instance of \ammFVS, on which all of the reduction rules
have been applied exhaustively,
thus (i) $\forall u \in U$, $\deg_{F \cup U} (u) \geq 2$,
and (ii) $F$ is an independent set.

Define $u \in U$ to be an \emph{interesting} vertex if $\deg_{F \cup U} (u) \geq 3$.
As already noted, $G[U]$ is a forest, all the leaves of which have an edge towards $F$,
otherwise Rule $3$ could still be applied.
Consider a root for each tree of $G[U]$.
For some tree $T$, let $v$ be an interesting vertex at maximum distance from the corresponding root,
i.e., $v$ is an interesting vertex of maximum depth.
Notice that such a tree cannot be an interesting path.
We branch on whether $v$ is in the feedback vertex set or not.
Towards this end, let $S' = S \cup \braces{v}$ and $F' = F \cup \braces{v}$,
while $\mathcal{I}_1 = (G, S', F, k)$ and $\mathcal{I}_2 = (G, S, F', k)$.
It holds that $\mathcal{I}$ is a Yes instance if and only if at least one of $\mathcal{I}_1, \mathcal{I}_2$ is a Yes instance,
while if $G[F']$ contains a cycle, $\mathcal{I}_2$ is a No instance and we discard it.
We replace $\mathcal{I}$ with the instances $\mathcal{I}_1$ and $\mathcal{I}_2$.

\begin{lemmarep}
    The branching strategy produces instances of reduced measure of progress,
    without reducing the number of good vertices.
    Additionally, whenever the branching places a vertex on the feedback vertex set,
    this vertex is good.
\end{lemmarep}

\begin{proof}
    Let $\mathcal{I} = (G, S, F, k)$ be an instance of {\ammFVS} and
    $\mathcal{I}_1 = (G, S', F, k)$, $\mathcal{I}_2 = (G, S, F', k)$ the instances produced by the branching strategy,
    where $S' = S \cup \braces{v}$ and $F' = F \cup \braces{v}$ for $v \in U$.
    Moreover, let $g, g_1, g_2$ denote the number of good vertices of each instance respectively.
    Recall that a vertex $s \in S$ is good if $\deg_F (s) \geq 2$ and $\deg_U (s) \leq 1$,
    and notice that $g \leq g_1$ and $g \leq g_2$.
    We assume that none of $\mathcal{I}_1, \mathcal{I}_2$ has been discarded, i.e., $G[F']$ is a forest.
    Notice that then, if $\deg_F(v) \geq 2$,
    it follows that $v$ has at least two neighbors in distinct connected components of $G[F]$.
    We will prove that $\mu(\mathcal{I}_1) < \mu(\mathcal{I})$ and
    $\mu(\mathcal{I}_2) < \mu(\mathcal{I})$.
    To this end, we consider three different cases:
    \begin{romanenumerate}
        \item $\deg_U (v) = 0$ and $\deg_F (v) \geq 3$, i.e.,
        $v$ is an isolated vertex of $G[U]$ with at least $3$ edges to $F$.
        On $\mathcal{I}_1$, it holds that $\mu(\mathcal{I}_1) \leq \mu(\mathcal{I}) - 1$,
        since $g_1 \geq g+1$.
        On the other hand, on $\mathcal{I}_2$, it holds that $\mu(\mathcal{I}_2) \leq \mu(\mathcal{I}) - 2$,
        since $\cc(G[F']) \leq \cc(G[F]) - 2$,
        otherwise $G[F']$ contains a cycle. 
        
        \item $\deg_U (v) = 1$ and $\deg_F (v) \geq 2$.
        On $\mathcal{I}_1$, it holds that $\mu(\mathcal{I}_1) \leq \mu(\mathcal{I}) - 1$,
        since $g_1 \geq g+1$.
        On the other hand, on $\mathcal{I}_2$, it holds that $\mu(\mathcal{I}_2) \leq \mu(\mathcal{I}) - 1$,
        since $\cc(G[F']) \leq \cc(G[F]) - 1$,
        otherwise $G[F']$ contains a cycle.
        As a matter of fact, the number of interesting paths might also increase.
        
        \item Lastly, either (a) $\deg_U (v) = 2$ and $\deg_F (v) \geq 1$,
        or (b) $\deg_U (v) \geq 3$.
        Since $v$ is an interesting vertex of maximum depth,
        for all of its descendants $w$ in its corresponding tree in $G[U]$
        it holds that $\deg_{F \cup U} (w) = 2$.
        On $\mathcal{I}_1$, for any child $u$ of $v$, it holds that $\deg_{V \setminus S'} (u) = 1$.
        In that case, by exhaustively applying Rule 3 and producing an instance $\mathcal{I}^*_1 = (G', S', F^*, k)$,
        it follows that $v$ has an additional edge to $F^*$ for each such child.
        In total, $v$ has at least $2$ edges towards $F^*$ in both (a) and (b),
        either due to its descendants or preexisting edges.
        Consequently, $g_1 \geq g+1$ and $\mu(\mathcal{I}_1) \leq \mu(\mathcal{I}) - 1$.
        Note that the number of interesting paths might also increase in the new instance.

        For $\mathcal{I}_2$, we consider (a) and (b) separately.
        \begin{itemize}
            \item In (a), $\cc(G[F']) \leq \cc(G[F])$ since $v$ has at least $1$ neighbor in $F$, while $p$ is increased by at least $1$.
            Indeed, since $v$ has at least one child $u$ in $U$, while $u$ and its descendants have degree $2$
            in $G[V \setminus S]$,
            it follows that the number of interesting paths is increased by at least 1.
        
            \item In (b), since $v$ does not necessarily have a neighbor in $F$, it holds that $\cc(G[F']) \leq \cc(G[F]) + 1$.
            However, $v$ has at least $2$ children in $U$, all of the descendants of which have degree $2$ in $G[V \setminus S]$,
            therefore the number of interesting paths $p$ is increased by at least $2$.
        \end{itemize}
        Consequently, $\mu(\mathcal{I}_2) \leq \mu(\mathcal{I}) - 1$.
    \end{romanenumerate}
    This completes the proof.
\end{proof}

\subparagraph*{Complexity.}
Starting from an instance $(G, k)$ of \mmFVS,
we produce a minimal feedback vertex set $S_0$ of $G$ in polynomial time.
If $|S_0| \geq k$, we are done.
Alternatively, we produce instances of {\ammFVS}
by guessing the intersection of $S_0$ with some minimal feedback vertex set of $G$ of size at least $k$.
% For every such instance, we employ a branching strategy which produces at most $2$ new instances of reduced measure of progress.
% In the worst case, all the produced instances have measure of progress at least $2$ and no further branching can be applied,
% i.e. they are path-restricted instances.
% 
% Let $\mathcal{I} = (G, S, F, k)$ be an instance of {\ammFVS} resulting from the initial guess of the partition of
% $S_0$ with the feedback vertex set and the forest.
Let $\mathcal{I} = (G, S, F, k)$ be one such instance.
It holds that $\mu(\mathcal{I}) \leq k + c$, where $c = \cc(G[F])$, therefore the branching will perform at most $k+c$ steps.
Notice that, at any step of the branching procedure, the number of good vertices never decreases.
Now, consider a path-restricted instance $\mathcal{I}' = (G', S', F', k)$ resulting from branching starting on $\mathcal{I}$,
on which branching, exactly $\ell$ times a vertex was placed in the feedback vertex set, therefore $|S'|-|S| = \ell$.
There are at most $\binom{k+c}{\ell}$ different such instances, each of which has at least $\ell$ good vertices,
thus \cref{thm:path_restricted_algo} requires time at most $3^{k-\ell} n^{\bO(1)}$. 
Since $0 \leq \ell \leq k+c$, and there are at most $\sum_{c=0}^k \binom{k}{c}$ different instances $\mathcal{I}$,
the algorithm runs in time $9.34^k n^{\bO(1)}$, as
\begin{align*}
    \sum_{c = 0}^k \binom{k}{c} \sum_{\ell = 0}^{k + c} \binom{k+c}{\ell} 3^{k-\ell} &=
    3^k \sum_{c = 0}^k \binom{k}{c} \sum_{\ell = 0}^{k + c} \binom{k+c}{\ell} 3^{-\ell} = 3^k \sum_{c = 0}^k \binom{k}{c} \parens*{\frac{4}{3}}^{k+c}\\
    &= 4^k \sum_{c = 0}^k \binom{k}{c} \parens*{\frac{4}{3}}^c = 4^k \parens*{\frac{7}{3}}^k \leq 9.34^k.
\end{align*}

\section{The Extension Problem}\label{sec:extension_fvs}

In this section we consider the following extension problem.

\problemdef{\textsc{Minimal FVS Extension}}
{A graph $G=(V,E)$ and a set $S\subseteq V$.}
{Determine whether there exists $S^* \supseteq S$ such that $S^*$ is
a minimal feedback vertex set of $G$.} 

Observe that this is a special case of \ammFVS,
since we essentially set $F = \varnothing$ and do not care about the size of the produced solution,
albeit with the difference that now we will not focus on the case where $V \setminus S$ is already acyclic.
This extension problem was already shown to be W[1]-hard parameterized by $|S|$ by Casel et al.~\cite{CaselFGMS22}.
One question that was left open, however, was whether it is solvable in polynomial time for fixed $|S|$, that is,
whether it belongs in the class XP.
Superficially, this seems somewhat surprising, because for the
closely related \textsc{Maximum Minimal Vertex Cover} and \textsc{Upper Dominating Set} problems,
membership of the extension problem in XP is almost trivial:
it suffices to guess for each $v \in S$ a private edge or vertex that is only dominated by $v$,
remove from consideration other vertices that dominate this private edge or vertex,
and then attempt to find any feasible solution.
The reason that this strategy does not seem to work for feedback vertex set is that for each $v \in S$
we would have to guess a private cycle.
Since a priori we have no bound on the length of such a cycle,
there is no obvious way to achieve this task in $n^{f(k)}$ time.

Though we do not settle the complexity of the extension problem for fixed $k$,
we provide evidence that obtaining a polynomial time algorithm would be a challenging task,
because it would imply a similar algorithm for the $k$-\textsc{in-a-Tree} problem.
In the latter, we are given a graph $G$ and a set $T$ of $k$ terminals and are asked to find
a set $T^*$ such that $T \subseteq T^*$ and $G[T^*]$ is a tree~\cite{ChudnovskyS10,LaiLT20}.

\begin{theorem}\label{thm:k_in_a_tree}
    {\sc $k$-in-a-Tree} parameterized by $k$ is fpt-reducible to
    {\sc Minimal FVS Extension} parameterized by the size of the given set.
\end{theorem}

\begin{proof}
    Consider an instance $G = (V, E)$ of $k$-\textsc{in-a-Tree}, with terminal set $T$.
    Let $T = \{t_1, \ldots, t_k\}$.
    We add to the graph $k-1$ new vertices, $s_1, \ldots, s_{k-1}$ and connect each $s_i$ to $t_i$ and to
    $t_{i+1}$, for $i \in [k-1]$.
    We set $S = \{s_1, \ldots, s_{k-1}\}$.
    This completes the construction.
    Clearly, this reduction preserves the value of the parameter.
    
    To see correctness, suppose first that a tree $T^* \supseteq T$ exists in $G$.
    We set $S_1 = S \cup (V \setminus T^*)$ in the new graph.
    $S_1$ is a feedback vertex set, because removing it from the graph leaves $T^*$, which is a tree.
    $S_1$ contains $S$. Furthermore, if $S_1$ is not minimal, we greedily remove
    from it arbitrary vertices until we obtain a minimal feedback vertex set $S_2$.
    We claim that $S_2$ must still contain $S$.
    Indeed, each vertex $s_i$, for $i \in [k-1]$, has a private cycle, since its neighbors $t_i,t_{i+1} \in T^*$.
    For the converse direction, if there exists in the new graph a minimal feedback
    vertex set $S^*$ that contains $S$, then the remaining forest $F^* = V \setminus
    S^*$ must contain $T$, since each vertex of $S$ must have a private cycle in
    the forest, and vertices of $S$ have degree $2$. Furthermore, all vertices of
    $T$ must be in the same component of $F^*$, because to obtain a private cycle
    for $s_i$, we must have a path from $t_i$ to $t_{i+1}$ in $F^*$, for all
    $i \in [k-1]$. Therefore, in this case we have found an induced tree in $G$ that
    contains all terminals.
\end{proof}

\section{Conclusions and Open Problems}

We have precisely determined the complexity of {\mmFVS} with respect to
structural parameters from vertex cover to treewidth as being slightly
super-exponential. One natural question to consider would then be to examine if
the same complexity can be achieved when the problem is parameterized by
clique-width. Regarding the complexity of the extension problem for sets of
fixed size $k$, we have shown that this is at least as hard as the well-known
(and wide open) $k$-\textsc{in-a-Tree} problem. Barring a full resolution of
this question, it would also be interesting to ask if the converse reduction
also holds, which would prove that the two problems are actually equivalent.

\bibliography{bibliography}

\begin{thebibliography}{10}

\bibitem{AbouEishaHLMRZ18}
Hassan AbouEisha, Shahid Hussain, Vadim~V. Lozin, J{\'{e}}r{\^{o}}me Monnot, Bernard Ries, and Viktor Zamaraev.
\newblock Upper domination: Towards a dichotomy through boundary properties.
\newblock {\em Algorithmica}, 80(10):2799--2817, 2018.
\newblock \href {https://doi.org/10.1007/s00453-017-0346-9} {\path{doi:10.1007/s00453-017-0346-9}}.

\bibitem{AraujoBCS22}
J{\'{u}}lio Ara{\'{u}}jo, Marin Bougeret, Victor~A. Campos, and Ignasi Sau.
\newblock Introducing lop-kernels: {A} framework for kernelization lower bounds.
\newblock {\em Algorithmica}, 84(11):3365--3406, 2022.
\newblock \href {https://doi.org/10.1007/s00453-022-00979-z} {\path{doi:10.1007/s00453-022-00979-z}}.

\bibitem{AraujoBCS23}
J{\'{u}}lio Ara{\'{u}}jo, Marin Bougeret, Victor~A. Campos, and Ignasi Sau.
\newblock Parameterized complexity of computing maximum minimal blocking and hitting sets.
\newblock {\em Algorithmica}, 85(2):444--491, 2023.
\newblock \href {https://doi.org/10.1007/s00453-022-01036-5} {\path{doi:10.1007/s00453-022-01036-5}}.

\bibitem{BasteST20}
Julien Baste, Ignasi Sau, and Dimitrios~M. Thilikos.
\newblock A complexity dichotomy for hitting connected minors on bounded treewidth graphs: the chair and the banner draw the boundary.
\newblock In {\em Proceedings of the 2020 {ACM-SIAM} Symposium on Discrete Algorithms, {SODA} 2020}, pages 951--970. {SIAM}, 2020.
\newblock \href {https://doi.org/10.1137/1.9781611975994.57} {\path{doi:10.1137/1.9781611975994.57}}.

\bibitem{BazganBCFJKLLMP18}
Cristina Bazgan, Ljiljana Brankovic, Katrin Casel, Henning Fernau, Klaus Jansen, Kim{-}Manuel Klein, Michael Lampis, Mathieu Liedloff, J{\'{e}}r{\^{o}}me Monnot, and Vangelis~Th. Paschos.
\newblock The many facets of upper domination.
\newblock {\em Theor. Comput. Sci.}, 717:2--25, 2018.
\newblock \href {https://doi.org/10.1016/j.tcs.2017.05.042} {\path{doi:10.1016/j.tcs.2017.05.042}}.

\bibitem{BelmonteKLMO22}
R\'{e}my Belmonte, Eun~Jung Kim, Michael Lampis, Valia Mitsou, and Yota Otachi.
\newblock Grundy distinguishes treewidth from pathwidth.
\newblock {\em SIAM Journal on Discrete Mathematics}, 36(3):1761--1787, 2022.
\newblock \href {https://doi.org/10.1137/20M1385779} {\path{doi:10.1137/20M1385779}}.

\bibitem{BergougnouxBBK20}
Benjamin Bergougnoux, {\'{E}}douard Bonnet, Nick Brettell, and O{-}joung Kwon.
\newblock Close relatives of feedback vertex set without single-exponential algorithms parameterized by treewidth.
\newblock In {\em 15th International Symposium on Parameterized and Exact Computation, {IPEC} 2020}, volume 180 of {\em LIPIcs}, pages 3:1--3:17. Schloss Dagstuhl - Leibniz-Zentrum f{\"{u}}r Informatik, 2020.
\newblock \href {https://doi.org/10.4230/LIPIcs.IPEC.2020.3} {\path{doi:10.4230/LIPIcs.IPEC.2020.3}}.

\bibitem{BodlaenderCKN15}
Hans~L. Bodlaender, Marek Cygan, Stefan Kratsch, and Jesper Nederlof.
\newblock Deterministic single exponential time algorithms for connectivity problems parameterized by treewidth.
\newblock {\em Inf. Comput.}, 243:86--111, 2015.
\newblock \href {https://doi.org/10.1016/j.ic.2014.12.008} {\path{doi:10.1016/j.ic.2014.12.008}}.

\bibitem{BonamyKNPSW18}
Marthe Bonamy, Lukasz Kowalik, Jesper Nederlof, Michal Pilipczuk, Arkadiusz Socala, and Marcin Wrochna.
\newblock On directed feedback vertex set parameterized by treewidth.
\newblock In {\em Graph-Theoretic Concepts in Computer Science - 44th International Workshop, {WG} 2018}, volume 11159 of {\em Lecture Notes in Computer Science}, pages 65--78. Springer, 2018.
\newblock \href {https://doi.org/10.1007/978-3-030-00256-5\_6} {\path{doi:10.1007/978-3-030-00256-5\_6}}.

\bibitem{BonnetBKM19}
{\'{E}}douard Bonnet, Nick Brettell, O{-}joung Kwon, and D{\'{a}}niel Marx.
\newblock Generalized feedback vertex set problems on bounded-treewidth graphs: Chordality is the key to single-exponential parameterized algorithms.
\newblock {\em Algorithmica}, 81(10):3890--3935, 2019.
\newblock \href {https://doi.org/10.1007/s00453-019-00579-4} {\path{doi:10.1007/s00453-019-00579-4}}.

\bibitem{BonnetLP18}
{\'{E}}douard Bonnet, Michael Lampis, and Vangelis~Th. Paschos.
\newblock Time-approximation trade-offs for inapproximable problems.
\newblock {\em J. Comput. Syst. Sci.}, 92:171--180, 2018.
\newblock \href {https://doi.org/10.1016/j.jcss.2017.09.009} {\path{doi:10.1016/j.jcss.2017.09.009}}.

\bibitem{dam/BoriaCP15}
Nicolas Boria, Federico~Della Croce, and Vangelis~Th. Paschos.
\newblock On the max min vertex cover problem.
\newblock {\em Discret. Appl. Math.}, 196:62--71, 2015.
\newblock \href {https://doi.org/10.1016/j.dam.2014.06.001} {\path{doi:10.1016/j.dam.2014.06.001}}.

\bibitem{CaselFGMS22}
Katrin Casel, Henning Fernau, Mehdi~Khosravian Ghadikolaei, J{\'{e}}r{\^{o}}me Monnot, and Florian Sikora.
\newblock On the complexity of solution extension of optimization problems.
\newblock {\em Theor. Comput. Sci.}, 904:48--65, 2022.
\newblock \href {https://doi.org/10.1016/j.tcs.2021.10.017} {\path{doi:10.1016/j.tcs.2021.10.017}}.

\bibitem{isaac/ChakrabortyFMT24}
Dipayan Chakraborty, Florent Foucaud, Diptapriyo Majumdar, and Prafullkumar Tale.
\newblock Tight (double) exponential bounds for identification problems: Locating-dominating set and test cover.
\newblock In {\em 35th International Symposium on Algorithms and Computation, {ISAAC} 2024}, volume 322 of {\em LIPIcs}, pages 19:1--19:18. Schloss Dagstuhl - Leibniz-Zentrum f{\"{u}}r Informatik, 2024.
\newblock \href {https://doi.org/10.4230/LIPICS.ISAAC.2024.19} {\path{doi:10.4230/LIPICS.ISAAC.2024.19}}.

\bibitem{ChaudharyMP23}
Juhi Chaudhary, Sounaka Mishra, and B.~S. Panda.
\newblock Minimum maximal acyclic matching in proper interval graphs.
\newblock {\em Discret. Appl. Math.}, 360:414--427, 2025.
\newblock \href {https://doi.org/10.1016/j.dam.2024.10.012} {\path{doi:10.1016/j.dam.2024.10.012}}.

\bibitem{ChudnovskyS10}
Maria Chudnovsky and Paul~D. Seymour.
\newblock The three-in-a-tree problem.
\newblock {\em Comb.}, 30(4):387--417, 2010.
\newblock \href {https://doi.org/10.1007/s00493-010-2334-4} {\path{doi:10.1007/s00493-010-2334-4}}.

\bibitem{CyganFKLMPPS15}
Marek Cygan, Fedor~V. Fomin, Lukasz Kowalik, Daniel Lokshtanov, D{\'{a}}niel Marx, Marcin Pilipczuk, Michal Pilipczuk, and Saket Saurabh.
\newblock {\em Parameterized Algorithms}.
\newblock Springer, 2015.
\newblock \href {https://doi.org/10.1007/978-3-319-21275-3} {\path{doi:10.1007/978-3-319-21275-3}}.

\bibitem{CyganNPPRW22}
Marek Cygan, Jesper Nederlof, Marcin Pilipczuk, Michal Pilipczuk, Johan M.~M. van Rooij, and Jakub~Onufry Wojtaszczyk.
\newblock Solving connectivity problems parameterized by treewidth in single exponential time.
\newblock {\em {ACM} Trans. Algorithms}, 18(2):17:1--17:31, 2022.
\newblock \href {https://doi.org/10.1145/3506707} {\path{doi:10.1145/3506707}}.

\bibitem{Diestel17}
Reinhard Diestel.
\newblock {\em Graph Theory}, volume 173 of {\em Graduate texts in mathematics}.
\newblock Springer, 2017.
\newblock \href {https://doi.org/10.1007/978-3-662-53622-3} {\path{doi:10.1007/978-3-662-53622-3}}.

\bibitem{DuarteEHKKLPSS21}
Gabriel~L. Duarte, Hiroshi Eto, Tesshu Hanaka, Yasuaki Kobayashi, Yusuke Kobayashi, Daniel Lokshtanov, Lehilton L.~C. Pedrosa, Rafael C.~S. Schouery, and U{\'{e}}verton~S. Souza.
\newblock Computing the largest bond and the maximum connected cut of a graph.
\newblock {\em Algorithmica}, 83(5):1421--1458, 2021.
\newblock \href {https://doi.org/10.1007/s00453-020-00789-1} {\path{doi:10.1007/s00453-020-00789-1}}.

\bibitem{jcss/DubloisHGLM22}
Louis Dublois, Tesshu Hanaka, Mehdi {Khosravian Ghadikolaei}, Michael Lampis, and Nikolaos Melissinos.
\newblock (in)approximability of maximum minimal {FVS}.
\newblock {\em J. Comput. Syst. Sci.}, 124:26--40, 2022.
\newblock \href {https://doi.org/10.1016/j.jcss.2021.09.001} {\path{doi:10.1016/j.jcss.2021.09.001}}.

\bibitem{DubloisLP22}
Louis Dublois, Michael Lampis, and Vangelis~Th. Paschos.
\newblock Upper dominating set: Tight algorithms for pathwidth and sub-exponential approximation.
\newblock {\em Theor. Comput. Sci.}, 923:271--291, 2022.
\newblock \href {https://doi.org/10.1016/j.tcs.2022.05.013} {\path{doi:10.1016/j.tcs.2022.05.013}}.

\bibitem{icalp/FoucaudGK0IST24}
Florent Foucaud, Esther Galby, Liana Khazaliya, Shaohua Li, Fionn~Mc Inerney, Roohani Sharma, and Prafullkumar Tale.
\newblock Problems in {NP} can admit double-exponential lower bounds when parameterized by treewidth or vertex cover.
\newblock In {\em 51st International Colloquium on Automata, Languages, and Programming, {ICALP} 2024}, volume 297 of {\em LIPIcs}, pages 66:1--66:19. Schloss Dagstuhl - Leibniz-Zentrum f{\"{u}}r Informatik, 2024.
\newblock \href {https://doi.org/10.4230/LIPICS.ICALP.2024.66} {\path{doi:10.4230/LIPICS.ICALP.2024.66}}.

\bibitem{FuriniLS17}
Fabio Furini, Ivana Ljubic, and Markus Sinnl.
\newblock An effective dynamic programming algorithm for the minimum-cost maximal knapsack packing problem.
\newblock {\em Eur. J. Oper. Res.}, 262(2):438--448, 2017.
\newblock \href {https://doi.org/10.1016/j.ejor.2017.03.061} {\path{doi:10.1016/j.ejor.2017.03.061}}.

\bibitem{arxiv/GaikwadKMST22}
Ajinkya Gaikwad, Hitendra Kumar, Soumen Maity, Saket Saurabh, and Shuvam~Kant Tripathi.
\newblock Maximum minimal feedback vertex set: {A} parameterized perspective, 2022.
\newblock \href {https://arxiv.org/abs/2208.01953} {\path{arXiv:2208.01953}}.

\bibitem{GourvesMP13}
Laurent Gourv{\`{e}}s, J{\'{e}}r{\^{o}}me Monnot, and Aris Pagourtzis.
\newblock The lazy bureaucrat problem with common arrivals and deadlines: Approximation and mechanism design.
\newblock In {\em Fundamentals of Computation Theory - 19th International Symposium, {FCT} 2013}, volume 8070 of {\em Lecture Notes in Computer Science}, pages 171--182. Springer, 2013.
\newblock \href {https://doi.org/10.1007/978-3-642-40164-0\_18} {\path{doi:10.1007/978-3-642-40164-0\_18}}.

\bibitem{books/aw/GKP1994}
Ronald~L. Graham, Donald~E. Knuth, and Oren Patashnik.
\newblock {\em Concrete Mathematics: {A} Foundation for Computer Science, 2nd Ed}.
\newblock Addison-Wesley, 1994.

\bibitem{HanakaKKY21}
Tesshu Hanaka, Yasuaki Kobayashi, Yusuke Kobayashi, and Tsuyoshi Yagita.
\newblock Finding a maximum minimal separator: Graph classes and fixed-parameter tractability.
\newblock {\em Theor. Comput. Sci.}, 865:131--140, 2021.
\newblock \href {https://doi.org/10.1016/j.tcs.2021.03.006} {\path{doi:10.1016/j.tcs.2021.03.006}}.

\bibitem{HarutyunyanLM21}
Ararat Harutyunyan, Michael Lampis, and Nikolaos Melissinos.
\newblock Digraph coloring and distance to acyclicity.
\newblock {\em Theory Comput. Syst.}, 68(4):986--1013, 2024.
\newblock \href {https://doi.org/10.1007/S00224-022-10103-X} {\path{doi:10.1007/S00224-022-10103-X}}.

\bibitem{jcss/ImpagliazzoPZ01}
Russell Impagliazzo, Ramamohan Paturi, and Francis Zane.
\newblock Which problems have strongly exponential complexity?
\newblock {\em J. Comput. Syst. Sci.}, 63(4):512--530, 2001.
\newblock \href {https://doi.org/10.1006/jcss.2001.1774} {\path{doi:10.1006/jcss.2001.1774}}.

\bibitem{LaiLT20}
Kai{-}Yuan Lai, Hsueh{-}I Lu, and Mikkel Thorup.
\newblock Three-in-a-tree in near linear time.
\newblock In {\em Proceedings of the 52nd Annual {ACM} {SIGACT} Symposium on Theory of Computing, {STOC} 2020}, pages 1279--1292. {ACM}, 2020.
\newblock \href {https://doi.org/10.1145/3357713.3384235} {\path{doi:10.1145/3357713.3384235}}.

\bibitem{Lampis21}
Michael Lampis.
\newblock Minimum stable cut and treewidth.
\newblock In {\em 48th International Colloquium on Automata, Languages, and Programming, {ICALP} 2021}, volume 198 of {\em LIPIcs}, pages 92:1--92:16. Schloss Dagstuhl - Leibniz-Zentrum f{\"{u}}r Informatik, 2021.
\newblock \href {https://doi.org/10.4230/LIPIcs.ICALP.2021.92} {\path{doi:10.4230/LIPIcs.ICALP.2021.92}}.

\bibitem{mfcs/LampisMV23}
Michael Lampis, Nikolaos Melissinos, and Manolis Vasilakis.
\newblock Parameterized max min feedback vertex set.
\newblock In {\em 48th International Symposium on Mathematical Foundations of Computer Science, {MFCS} 2023}, volume 272 of {\em LIPIcs}, pages 62:1--62:15. Schloss Dagstuhl - Leibniz-Zentrum f{\"{u}}r Informatik, 2023.
\newblock \href {https://doi.org/10.4230/LIPIcs.MFCS.2023.62} {\path{doi:10.4230/LIPIcs.MFCS.2023.62}}.

\bibitem{LokshtanovMS18}
Daniel Lokshtanov, D{\'{a}}niel Marx, and Saket Saurabh.
\newblock Slightly superexponential parameterized problems.
\newblock {\em {SIAM} J. Comput.}, 47(3):675--702, 2018.
\newblock \href {https://doi.org/10.1137/16M1104834} {\path{doi:10.1137/16M1104834}}.

\bibitem{MishraS01}
Sounaka Mishra and Kripasindhu Sikdar.
\newblock On the hardness of approximating some np-optimization problems related to minimum linear ordering problem.
\newblock {\em {RAIRO} Theor. Informatics Appl.}, 35(3):287--309, 2001.
\newblock \href {https://doi.org/10.1051/ita:2001121} {\path{doi:10.1051/ita:2001121}}.

\bibitem{Pilipczuk11}
Michal Pilipczuk.
\newblock Problems parameterized by treewidth tractable in single exponential time: {A} logical approach.
\newblock In {\em Mathematical Foundations of Computer Science 2011 - 36th International Symposium, {MFCS} 2011,}, volume 6907 of {\em Lecture Notes in Computer Science}, pages 520--531. Springer, 2011.
\newblock \href {https://doi.org/10.1007/978-3-642-22993-0\_47} {\path{doi:10.1007/978-3-642-22993-0\_47}}.

\bibitem{siamdm/Zehavi17}
Meirav Zehavi.
\newblock Maximum minimal vertex cover parameterized by vertex cover.
\newblock {\em {SIAM} J. Discret. Math.}, 31(4):2440--2456, 2017.
\newblock \href {https://doi.org/10.1137/16M109017X} {\path{doi:10.1137/16M109017X}}.

\end{thebibliography}

\end{document}